\documentclass[12pt,oneside,english]{article}

\usepackage{epsfig,color,mathtools,url,natbib,amstext,color,comment,paralist,amsthm,amsmath,amssymb,setspace,nicefrac,booktabs,multirow,soul,babel,fullpage}


\usepackage{titlesec}
\titleformat*{\section}{\large\bfseries}

\newcommand{\subtau}{\hat{\tau}}
\newcommand{\subbeta}{\hat{\beta}}
\newcommand{\hatb}{\hat{\beta}}

\newcommand{\by}{\underline{y}}

\newcommand{\RR}{\mathbb{R}}
\newcommand{\EE}{\mathbb{E}}
\newcommand{\PP}{\mathbb{P}}

\DeclareMathOperator{\sgn}{sgn}

\newcommand{\ind}[1]{\mathbf{1}_{#1}}

\newcommand{\comments}[1]{}

\newtheorem{lemma}{Lemma}
\newtheorem{theorem}{Theorem}
\newtheorem{corollary}{Corollary}

\theoremstyle{definition}
\newtheorem{example}{Example}

\theoremstyle{remark} 
\newtheorem{remark}{Remark}

\title{Identifying Present-Bias from the Timing of Choices\footnote{We thank Ned Augenblick, Stefano DellaVigna, Ori Heffetz, Botond K\H{o}szegi, Muriel Niederle, Charles Sprenger, Dmitry Taubinsky, and Florian Zimmermann for insightful and encouraging comments. Part of the work on this paper was carried out while the authors visited {\it briq}, whose hospitality is gratefully acknowledged.}}
\date{\today}
\author{Paul Heidhues\\DICE \and Philipp Strack \\UC Berkeley}

\begin{document}

\maketitle

\begin{abstract}
Timing decisions are common: when to file your taxes, finish a referee report, or complete a task at work. We ask whether time preferences can be inferred when \textsl{only} task completion is observed. To answer this question, we analyze the following model: each period a decision maker faces the choice whether to complete the task today or to postpone it to later. Cost and benefits of task completion cannot be directly observed by the analyst, but the analyst knows that net benefits are drawn independently between periods from a time-invariant distribution and that the agent has time-separable utility. Furthermore, we suppose
the analyst can observe the agent's exact stopping probability. We establish that for any agent with quasi-hyperbolic $\beta,\delta$-preferences and given level of partial naivete $\hat{\beta}$, the probability of completing the task conditional on not having done it earlier increases towards the deadline. And conversely, for any given preference parameters $\beta,\delta$ and (weakly increasing) profile of task completion probability, there exists a stationary payoff distribution that rationalizes her behavior as long as the agent is either sophisticated or fully naive. An immediate corollary being that, without parametric assumptions, it is impossible to rule out time-consistency even when imposing an a priori assumption on the permissible long-run discount factor. We also provide an exact partial identification result when the analyst can, in addition to the stopping probability, observe the agent's continuation value. 

\end{abstract}

\pagebreak

\setcounter{page}{1}

\section{Introduction}

Intuition and evidence suggests that many individuals are time-inconsistent; at any particular point in time the (near) present gets an additional weight in intertemporal tradeoffs \citep[e.g.][]{strotz,frederickloewenstein2,augenblickniederle,augenblickrabin}. Especially when individuals fail to fully anticipate their predictable preference changes, such present-focused individuals tend to procrastinate \citep{akerlofobedience,odonoghuerabindoing,odonoghuerabinchoice}: they will often excessively delay the completion of tedious tasks such as filing taxes or paying parking-tickets. And when facing a gratifying  task---such as taking a day off---, present-focused individuals often precrastinate. To model the resulting interpersonal-conflict of preference changes in a simple and tractable way, \cite{laibson} adopted intergenerational discounting models \citep{phelpspollak} to individual decision-making. His quasi-hyperbolic discounting model captures the present-focus of individuals by introducing an additional present-bias parameter that discounts all future utility into \cite{samuelsondu}'s time-separable exponential-discounting model. \cite{odonoghuerabindoing,odonoghuerabinchoice} extend this framework by introducing (partial) naivete, and illustrating such individuals' tendency to delay unpleasent tasks. Since excessive procrastination is a robust prediction of (naive) hyperbolic discounting models, it seems natural to use task-completion data to identify time-inconsistent preferences from the pattern of completion times. In line with this idea, previous research classifies individuals as time-inconsistent if they complete tasks at or close to the deadline \citep{brownprevitero,frakeswasserman} or  estimates the degree of time-inconsistency from completion times under parametric assumptions \citep{martinezmeier}.\footnote{\cite{brownprevitero} classify individuals that select their health care plan close to the deadline as procrastinators and look for correlated behavior in other financial domains. \cite{frakeswasserman} investigate the behavior of patent officers that have to complete a given quota of applications supposing that the cost of working on a patent are deterministic and identical across days. In their model, for conventional discount rates the empirically observed bunching close to the deadline is inconsistent with exponential discounting. While earlier papers do not address the concern of unobservable and random opportunity cost, \cite{martinezmeier} allow for random opportunity costs and use a parametric approach to identify time preferences.}

In this paper, we ask whether time preferences can be inferred by an outside observer---referred to as the analyst---when \textsl{only} task completion is observed \textsl{absent} parametric assumptions on the (unobservable) cost and benefit of task completion. A key difficulty in doing so is to separate naivete or time-preference-based explanations of delay from those due to the option value of waiting \citep{wald,weisbrod,dixitpindyck}: whenever the cost of doing a certain task is stochastic, a time-consistent individual may wait in the hope of getting a lower cost draw tomorrow.\footnote{Throughout, we abstract from another reason that tasks may not be completed: forgetting. Conceptually, one can think of the agent in our analysis as getting a non-intrusive reminder at the beginning of every period. This is not to say that limited memory and the strategic response to it are unimportant in determining task completion behavior in the field. See, for example, \cite{heffetzodonoghue} for how reminders determine when parking fines are payed, \cite{altmannetraxler} for how deadlines and reminders determine the probability of making a check-up appointment at the dentist, and \cite{ericson2} for how time-inconsistency and limited memory interact. }

Section \ref{sec:setup} introduces our task-completion model. We consider an analyst who, from observing task completion times of a partially-naive quasi-hyperbolic discounter, tries to learn about some or all of the following parameters: the long-run discount factor $\delta$, the present-bias parameter $\beta$, or the degree of sophistication $\subbeta$. To facilitate learning by the analyst, we assume that the agent's task-completion payoffs are drawn each period from the same underlying payoff distribution. Absent any such a priori restriction, it is straightforward to rationalize any observed stopping behavior independently of the agent's taste for immediate gratification and degree of sophistication, leaving no hope for identification thereof.\footnote{For example, suppose in every period the cost of doing the task is either one or zero, allowing for time-varying probability that the cost are zero. Simply setting the probability that the cost are zero in each period equal to that period's observed task completion probability rationalizes the data for any time-separable utility function.} Furthermore, to make identification easier, we suppose that the analyst can observe the individual's exact stopping probability in each period. Intuitively, one may think of the analyst as having access to an ideal data set with (infinitely) many observations of either the same individual in identical situations or a homogenous group of individuals. Again, this assumption strongly favors the analyst's ability to learn about underlying parameters. Finally, we impose that individuals can be described as (partially) naive quasi-hyperbolic discounters. We are agnostic as to the nature of the task, so our analysis applies when task-completion leads to  immediate benefits, immediate costs, or both.

In Section \ref{sec:example}, we introduce two motivating examples. The first highlights that, even when the parametric form of the underlying unobservable payoff distribution are known, bunching at the deadline is insufficient to distinguish a time-consistent from a time-inconsistent agent. In the example, the cost of completing the task are drawn from a log-normal distribution and in every period the stopping behavior of time-consistent agent looks almost identical to that of an agent with a present-bias parameter $\beta = 0.7$, whose cost are drawn from a different log-normal distribution. The second example illustrates how the estimated present-bias can depend crucially on common parametric assumptions about the unobservable payoff distribution---even when the analyst knows (or guesses correctly) the long-run discount factor, as well as the mean and variance of the underlying stationary payoff distribution. While we suppose that in reality payoffs are drawn from a uniform distribution and the agent is time-consistent $\beta=\subbeta=1$, when the analyst supposes costs are drawn either from a normal, log-normal, extreme value, or logistic distribution, her  squared-distance-minimizing or likelihood-maximizing estimate of $\beta$ varies between $0.561 - 0.819$, with the exact value depending on the parametric family (and the degree of sophistication) the analyst imposes. 
Furthermore, the squared error associated with some of these incorrect estimates is below $0.232$\%---suggesting that with finite noisy data it is difficult for the analyst to realize when she picks an incorrect functional form.
Motivated by the importance of the parametric assumptions in the example, we turn to the main focus of the paper: what lessons about time-inconsistent preferences and naivete thereof can be learned non-parametrically?

As a useful preliminary step, Section \ref{sec:recursive} establishes that the agent's perceived continuation value is characterized by a simple recursive equation. Section \ref{sec:task_completition} establishes that for any quasi-hyperbolic discounter---independently of whether she is sophisticated or (partially) naive and of her degree of impatience---the subjective continuation value decreases the closer the agent gets to the deadline. To see the intuition behind the theorem, consider first the case in which the task always generates a net benefit. Then from the perspective of Self 1, all future selves are too impatient, and hence tend to perform the task to early. By extending the deadline, the formerly last period's self now can decide and perform the task later. As from any earlier self's perspective she is too eager to complete the task,  the direct effect of additional delay on any earlier self is positive. Now consider the former penultimate self; her perceived continuation value of waiting increases because she strictly prefers future selves to wait whenever they choose to do so. This, in turn, induces her to act more patiently, benefiting all earlier selfs, and so forth. Hence, in the case of net benefits, a quasi-hyperbolic discounter does not want to impose an earlier deadline. 

Consider next the case in which completing the task is always costly. When comparing a $(T-1)$-period to $T$-period deadline, Self 1 realizes that if she does not engage in the task in the $T$ period problem, Self 2 will face a $T-1$-period problem. That subgame is identical to the one she faces in the $T-1$ period problem, and future selves who are $s$ periods away from the deadline will therefore behave identically in the two problems. Hence for $s \in {1,\cdots T-1}$, the task completion probability $s$-periods before the deadline is identical, and due to discounting of future costs, Self 1 is strictly better off when selecting the $T$-period problem and not doing the task. The formal proof extends these intuitions to the case in which the support of the net benefit distribution can contain positive and negative payoffs.


 Because the agent in our model completes the task when the current benefit is greater than her subjective continuation value, Theorem \ref{prop:monotone-values} implies that a quasi-hyperbolic discounter becomes more and more likely to complete the task the closer she is to the deadline. This, therefore, provides another simple testable prediction, which also implies that the agent never wants to impose a shorter deadline.\footnote{Despite her tendency to procrastinate, hence, when the payoffs are independently drawn from a stationary distribution, a quasi-hyperbolic discounter's willingness to pay for an earlier deadline is always non-positive. This is noteworthy as self-imposed deadlines by students has been used to identify sophisticated procrastinators  \citep[e.g.][]{arielywertenbroch,bisinhyndman}; our result suggests that these students either do not have quasi-hyperbolic preferences or that they must foresee a non-stationary environment,
which induces them to impose an earlier deadline. A self-imposed-deadline-based classification, hence, is conservative in identifying agents who are aware of their time-inconsistent preferences.} Through a simple counterexample, however, we also highlight that this result relies on payoffs each period being drawn from the same underlying distribution.\footnote{Furthermore, in Section \ref{sec:discussion} we note that the prediction need not hold for a heterogenous population of time-consistent individuals that each faces a stationary payoff distribution.}

Section \ref{sec:unidentified} establishes our main result: if the agent is either sophisticated ($\hatb = \beta$) or fully naive ($\hat{\beta} =1$), for {\it any given} long-run discount factor $\delta$ and present-bias parameter $\beta$, any given penalty of not completing the task, and \textsl{any} weakly increasing profile of task completion, there exists a stationary payoff distribution that rationalizes the agent's behavior (Theorems \ref{thm:non-identifiability-sophisticate} and \ref{thm:non-identifiability}, respectively). This implies that for \textsl{any} data set the analyst may observe, absent parametric assumptions it is impossible for her to learn {\it anything} about the agent's degree of time-inconsistency or level of sophistication. Importantly, this absence of even partial identification continues to hold even if the analyst imposes a priori restrictions on permissible long-run discount factors. 
A very rough intuition for this fact is as follows: whether a self prefers to do a task today or tomorrow depends on her time preferences and on the perceived option value of waiting. The option value of waiting, in turn, depends on the payoff distribution. Through changing the unobservable payoff distribution, we can hence undo a change in the present bias or long-run discount factor of the agent. 

Technically, however, a local change in the payoff distribution changes continuation values in every period in a highly non-linear way, so to establish that we can construct an appropriate payoff distribution, we need a non-local argument. This is where we use the assumption that the agent is either sophisticated or fully naive.  The fact that a sophisticated agent makes no forecast error enables us rewrite the recursive equations determining the perceived continuation values in a simple manner. Based on this rewrite, we transfer the search for an appropriate distribution to that of solving for a fixed point of a  system of linear equations. This proof method, however, cannot be used if the agent is partially naive as the corresponding system becomes non-linear. 

For  a fully naive agent the problem becomes tractable for a different reason. Because a fully naive agent believes to be time-consistent, we can establish that a first-order stochastic increase in the stationary payoff distribution, increases the agent's subjective continuation value in every period (Lemma \ref{lem:aux-properties-naive}). In addition, we establish that we can map subjective continuation values into a payoff distribution that gives rise to the desired completion times in such a way that greater subjective continuation values lead to a first-oder stochastic increase in the stationary payoff distribution. The combination of these two steps leads to a monotone operator on subjective continuation values to which we can apply Tarski's Theorem, and thereby establish the existence of a payoff distribution that gives rise to the data's stopping probabilities. We also, however, provide a simple example in which a first-order stochastic dominance increase in the stationary payoff distribution makes a sophisticated quasi-hyperbolic discounter worse off. In the example, the agent prefers to pay a fixed utility-tax immediately upon completing the task. This tax reduces his temptation to stop even after a low payoff realization, and the induced more virtuous behavior of future selves overcompensates the direct payoff loss due to the tax. The example highlights why our proof technique does not cover the more general case of a partially naive agent.

In our proofs of Theorems \ref{thm:non-identifiability-sophisticate} and \ref{thm:non-identifiability}, we freely construct a stationary net-benefit distribution. One may hope to identify present-bias through economically meaningful restrictions on this distribution. Arguably, the most natural assumptions are those regarding the moments of the net-benefit distribution; for example, an analyst may have an idea regarding the possible expected net benefit of doing the task---that is regarding the mean of $F$---or may be willing to impose that net benefits do not vary to much between periods (restricting the variance of $F$). Our example in Section \ref{sec:example}, however, already highlights that even fixing these moments, common parametric assumptions can lead to widely varying estimates of the agent's time preferences. To expand on this point, in Section \ref{subsec:moments} we establish that as long as the penalty is unobservable or the task is mandatory, we can find a net benefit distribution with \textsl{any} given mean and non-zero variance that rationalizes the observed stopping behavior for a time-consistent agent with $\delta =1$. Any identification of present-bias parameter $\beta$ in this case, therefore, must follow from parametric restrictions on higher-order moments of the distribution, for which we see no convincing economic motivation in most contexts.

Section \ref{sec:rich data} asks whether non-parametric identification is feasible with richer data in which the analyst does not only observe the stopping probabilities but, in addition, observes the agent's willingness to pay for continuing with the stopping problem in each period. In the case of tax-filing, for example, this amount to eliciting the willingness to pay for having someone else file one's taxes immediately with zero hassle.\footnote{As we explain carefully in Section \ref{sec:rich data}, our procedure does not explicitly or implicitly rely on the agent comparing monetary rewards at different points in time, so it is robust to standard critiques of eliciting time-preference via monetary rewards \citep{augenblickniederle,ericsonlaibsonreview,ramsey}.} For the case of a sophisticated agent who's contemporaneous utility function is quasi-linear in money, Theorem \ref{thm:non-para_identification} provides an analytical answer in closed form. Indeed, to check whether or not the data is consistent with a given pair of parameters $\beta,\delta$,  the analyst only needs to verify a simple set of inequalities. The key analytical insight is contained in Lemma \ref{lem:mass_points_sufficient}, which establishes that it suffices to consider distributions that have $T+1$ mass points. Intuitively, the option value of waiting is determined by the probability with which the agent stops at given future point in time and the expected payoff conditional on doing so. Hence, moving the probability mass between any two continuation values to the expected payoff conditional on falling between these two values leaves the agent's continuation values and stopping probabilities unaltered. Therefore, the analyst can restrict attention to such relatively simple distributions. 
Economically, observing the continuation values allows the analyst to distinguish between a taste for immediate gratification and option-value-of-waiting-based delays because a high option value requires the unobservable payoffs to differ significantly. As a consequence, as the deadline approaches and the agent foresees less future draws, the option value must decrease quickly. In contrast, a present-biased agent's continuation value decreases at a slower rate. We also argue that at the cost of relying on numerical techniques commonly used in applied work, our set-identification result can be extended  straightforwardly to cover partial naivete and non-linear utility in money. 

Applying our Theorem \ref{thm:non-para_identification} to the example introduced in Section \ref{sec:example}, however, illustrates that the analyst may need to observe a large number of continuation values to be able to tightly identify the present-bias parameter. In the example, there is no meaningful identification with $5$ periods of data, but $20$ periods are enough to tightly identify $\beta$ when $\delta=1$ is known to the analyst. Given that we made a number of assumptions facilitating identification---such as that the exact stopping probabilities and continuation values are observable to the analyst---, we think that the overall message of our analysis suggests a substantial amount of additional data is needed to empirically identify a taste for immediate gratification or the degree of sophistication without relying on parametric assumptions. We point out that if the analyst observes a heterogenous population, much richer stopping patterns can be explained in Section \ref{sec:discussion}, where we conclude by discussing some broader implications of our analysis.



\section{Setup}
\label{sec:setup}
Let time $t=1,2,\cdots, T+1$ be discrete. We consider an agent with quasi-hyperbolic preferences who can choose when and whether to complete a single given task before some deadline $T$. More precisely, we suppose that the agents' utility is time-separable, and denote a level of instantaneous utility the agent receives in period $t$ by $u_t$; let 
\begin{equation}
\label{eq:true_pref}
U^t = u_t + \beta \, \sum_{s=t+1}^{T+1} \delta ^{s-t} \, u_s,
\end{equation}
denote the utility over sequence of $(u_t,\cdots u_{T+1})$ of self $t$. Following \cite{odonoghuerabindoing}, we allow the agent to have incorrect beliefs regarding future selves' behavior. The agent believes that all future selfs $r>t$ maximize 
\begin{equation}
\label{eq:sub_pref}
\hat{U}^r = u_r + \subbeta \,  \sum_{s=r+1}^{T+1} \delta ^{s-r} \, u_s.
\end{equation}
We allow for any vector of preference and belief parameters $(\delta,\beta,\subbeta) \in (0,1]^3$. In case $\subbeta =\beta =1$, the agent has time-consistent preferences with an exponential discount factor $\delta$. In case $\beta < 1$, she has a taste for immediate gratification.  We say she is sophisticated---i.e. perfectly predicts her future behavior---when $\subbeta =\beta$,  she is fully naive---i.e. believes that her future selves behave according to her current preference---if $\subbeta =1$, and otherwise say that she is partially naive. Our setup covers the case in which the agent overestimates her own future taste for immediate gratification $\subbeta<\beta$ as well as the case in which she underestimate it $\subbeta>\beta$. 

The agent can complete the task once during the periods $t=1,\cdots,T$, so that $T$ is the deadline before which the task needs to be completed. If the agent does not undertake the task in a given period $t=1,\cdots,T$, we normalize her instantaneous utility $u_t$ to zero. If she completes the task she gets an instantaneous utility of zero in period $T+1$, while if she did not complete the task by the end of period $T$, the agents gets a (utility) penalty of $\by/ (\beta \delta) \in \RR_- \cup \{-\infty\}$ in period $T+1$.\footnote{In other words, $\by$ is self $T$'s continuation value when not completing the task. Expressing the penalty in this way simplifies the exposition below.} Setting $\by=-\infty$, this encompasses the case where the task is \textsl{mandatory} so that  the agent is forced to complete the task by the end of period $T$; and setting $\by=0$, this encompasses the case in which the task is \textsl{optional} so the agent only completes the task if her active self decides to do so.  Finally, we suppose that in every period $t$ the instantaneous utility of completing the task is drawn independently from a given payoff distribution $F$, which is known to the agent. 

We look for \emph{perception-perfect equilibria} \citep{odonoghuerabindoing,odonoghuerabinchoice} in which each self $t$ chooses an optimal strategy given its prediction of future selves' behavior, and a self $t$'s prediction of future selves' behavior are consistent with how a future self with preference parameter $\subbeta$ would optimally behave. More formally, let $Y^t=(y_1, \cdots,y_t)$ be the history of payoff realizations up to time $t$. A pure strategy for Self $t$ is a mapping $\sigma_t(Y^{t-1},y_t) \rightarrow \{0,1\}$, with the interpretation that $1$ means Self $t$ completes the task. 
A perception-perfect equilibrium is a pair of strategies $(\sigma_1,\cdots,\sigma_T)$ 
and $(\hat{\sigma}_2, \cdots,\hat{\sigma}_T)$ such that for all
 $t \in \{1,\cdots,T\}$, $\sigma_t$ maximizes $U^t$ under the assumption that selves $r>t$ use strategy 
 $\hat{\sigma_r}$, and for all $t \in \{2,\cdots,T\}$, the strategy $\hat{\sigma}_t$ maximizes $\hat{U}^t$  under the assumption that selves $r>t$ use strategy $\hat{\sigma}_r$. In addition, we restrict attention to perception-perfect equilibria in which all selves that are indifferent between completing the task and waiting choose to wait.\footnote{Without a given tie-breaking assumption, we could rationalize any behavior by simply assuming that the payoff of completing the task is $0$ with certainty in all periods. In that case, any stopping probability in any period is trivially optimal, independently of the agent's time preferences. All our results below extend to the case in which the agent completes the task with some given positive probability when indifferent. Furthermore, in case the agent's benefit distribution admits a density, the tie-breaking assumption is obviously immaterial. And even otherwise, the case in which there is a mass-point at a payoff at which the agent is indifferent between completing the task and waiting is knife-edge.}

\section{Examples on the Influence of Parametric Assumptions}
\label{sec:example}

\begin{example}\label{ex:bar_plot}
To illustrate the difficulty of identifying time-inconsistency from an agent's stopping behavior, consider the following stylized example. A sophisticated agent receives a parking fine, which has to be paid within ten days of receiving it. In case she does not pay the fine, she incurs a known cost of $\$5$ in addition to the fine. Furthermore, the agent's long-run (daily) discount rate is (well approximated by) $\delta =1$.
\renewcommand{\baselinestretch}{1}

\begin{figure}
\begin{center}
\includegraphics[width=3.8in]{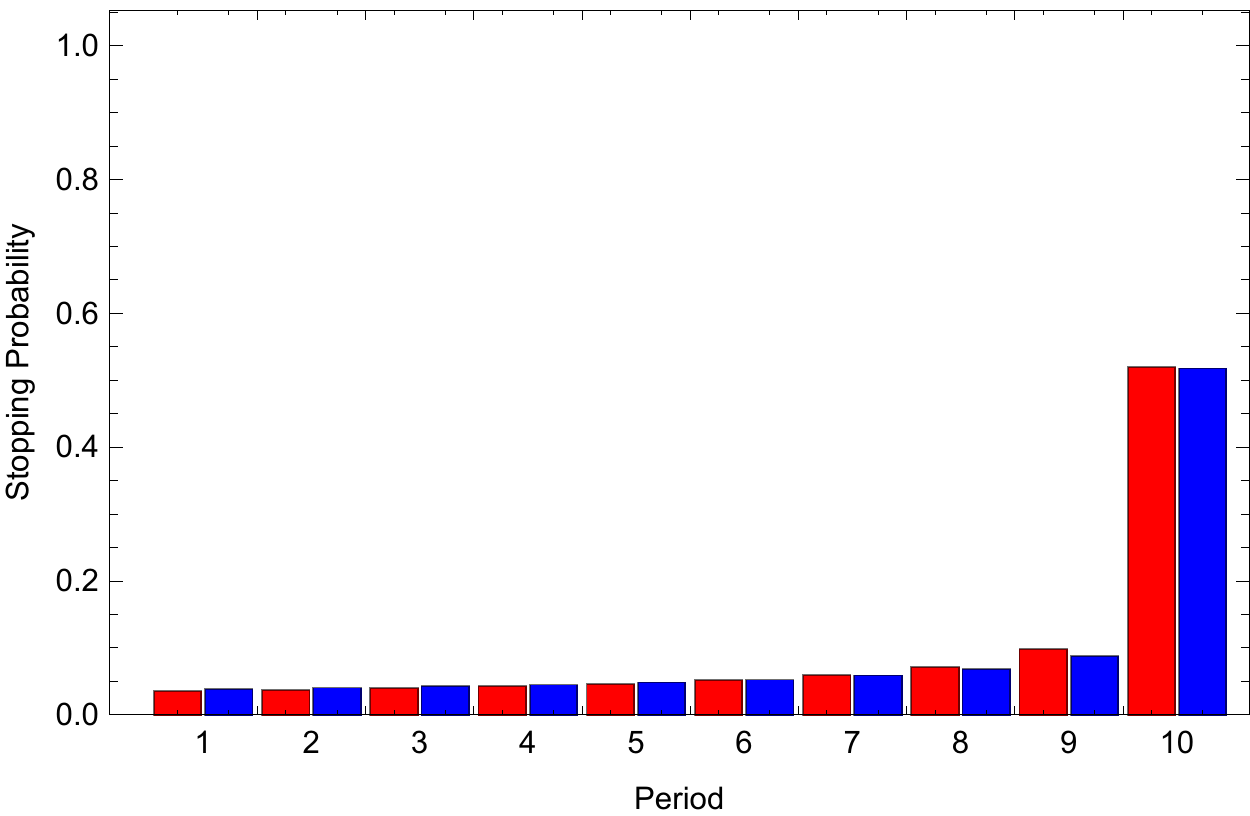}

\end{center}
\caption{Observed Task Completion Times. The above graphs illustrates the observed stopping times. In both cases $\delta =1$, and the penalty for not doing the task is $-5$. The red bar plot shows the distribution of task completion times of a time-consistent agent whose cost of completing the task are drawn from a log-normal distribution, whose underlying normal distribution has mean $\mu=1$ and variance $ \eta=1$. The blue bar plot that of a sophisticated time-inconsistent agent with $\beta = 0.7$ whose cost are drawn from a log-normal distribution with parameters $\mu=0, \eta=2.3$.} 
\label{fig:stopping_example}

\end{figure}

\renewcommand{\baselinestretch}{1.5}

Figure \ref{fig:stopping_example} compares the stopping behavior of a time consistent agent who draws the cost of completing the task from a log-normal distribution whose underlying normal distribution has mean $\mu =1$ and variance $\eta = 1$ (red bar plot) to that of a sophisticated time-inconsistent one with a present-bias parameter $\beta =0.7$ who draws the cost from  a log-normal distribution with parameters $\mu=0, \eta=2.3$ (blue bar plot).
\end{example}

An obvious first lesson from the example is that bunching at the deadline is no reliable guide to identifying time-inconsistency: both agents probability of completing the task in the final period is just above 50\%. Indeed, both agents stopping behavior is remarkably similar throughout and the observed stopping probabilities differ by less than $1\%$ in any period, suggesting that even an analyst who wants to test only between these two possible types faces a difficult problem in practice.\footnote{Independently of our work, \citeauthor*{heffetzodonoghue} observe that substantially different values of $\beta$ can explain the parking-ticket payment behavior in New York City, which they analyze in \cite{heffetzodonoghue}. They illustrate this supposing that the cost for paying the parking ticket is drawn from the small parametric family of distributions that has a mass point at zero and admit a constant density on an interval above zero. Their real-world data nicely demonstrates the practical importance of the identification challenge we illustrate in Example \ref{ex:bar_plot} with synthetic data.  We are very grateful to these authors for sharing their example with us during private communication.}


In the above illustrative example, the analyst knows or correctly guesses the parametric class of distributions (log-normal) from which the payoffs are drawn. The example suggests that without knowing its exact parameters, nevertheless, it is hard to correctly identify the time-preference parameters. In reality, however, payoffs are drawn from an unobservable payoff distribution and for typical field data---such as parking tickets---an analyst does \textsl{not} know the parametric form of the payoff distribution. The following example highlights how crucial common functional form assumptions routinely imposed in applied papers can be in determining the analyst's findings. For this example, we suppose that the analyst has precise prior knowledge about the mean and the variance of the unobservable payoff distributions but is unsure as to the exact parametric family from which these payoffs are drawn. Indeed, it strikes us as extremely unreasonable that an analyst has prior knowledge beyond some (typically vague) ideas about the first two moments of this distribution.

\begin{example} \label{ex:parametric form} We suppose that the agent has $5$ periods to complete the task and the agent's value of completing the task are drawn from a uniform distribution over $[-1,1]$; in reality the agent is time-consistent with $\beta=\delta=1$.\footnote{Think of a parent that promised their kid to see a theatre play that shows for seven more days. The parent is self-employed and needs to complete tasks at work as they come in. When not being very busy, the parent enjoys the joint activity. When very busy, however, he is distracted during the play and needs to stay up late afterwards completing his work tasks. Not going to the play after having promised to do so, however, is not a possibility.} The corresponding stopping probabilities are $0.25827, 0.304687, 0.375, 1/2, 1$, which we suppose the analyst can observe exactly. In addition, we assume the analyst knows the true mean ($0$) and standard deviation ($0.577$) of the stationary payoff distribution $F$ but not its exact functional form. Furthermore, suppose the analyst correctly imposes that $\delta =1$ when analyzing the data. Let the analyst consider four standard parametric families of distributions: normal, log-normal, extreme value, and logistic. For each of these families, the analyst selects the parameter $\beta$ that best fits---in the sense of squared distance or log-likelihood---the observed stopping probabilities allowing the agent to be either naive or sophisticated. Table \ref{tab:table1} reports the parameter estimates for $\beta$ and the squared distance/log-likelihood for the different parameterizations of the error distribution.\footnote{The estimates are computed using grid search with a distance of $0.0005$ between grid points.}

\begin{table}
  \begin{center}
    \label{tab:table1}
\begin{tabular}{@{}lllcll@{}} 
    \toprule
      \multirow{2}{*}{\textsl{Parametric Family}} & \multicolumn{2}{c}{\textsl{Sq. Distance Minimzation}} && \multicolumn{2}{c}{\textsl{Likelihood Maximization}} \\
      \cmidrule{2-3} \cmidrule{5-6} 
      &$\beta$ & Distance && $\beta$ &  Log-Likelihood  \\
      \hline
 \text{Normal Sophisticate} & 0.819 & 0.0026777 && 0.818 & 1.59188 \\
 \text{Normal Naive} & 0.817 & 0.00231803 && 0.816 & 1.59187 \\
 \text{Extreme Value Sophisticate} & 0.57 & 0.0402888 && 0.5705 & 1.59638 \\
 \text{Extreme Value Naive} & 0.561 & 0.0396802 && 0.562 & 1.59627 \\
 \text{Logistic Sophisticate} & 0.7605 & 0.00331235 && 0.7595 & 1.59189 \\
 \text{Logistic Naive} & 0.7565 & 0.00267175 && 0.7555 & 1.59188 \\
 \bottomrule
\end{tabular}
       \caption{Parameter estimates of $\beta$ and squared distance and log-likelihood.}
  \end{center}
\end{table}

The analyst's estimates of $\beta$ range between $0.561 - 0.819$ even in this idealized situation in which she has infinite data, actually knows the mean and standard deviation of $F$, and knows the long-run discount factor $\delta$. And if the analyst engaged in model testing selecting the model on the basis of minimizing squared distance or maximizing log-likelihood, she would conclude that the agent is naive time-inconsistent with $\beta =0.817/0.816$ while in truth the agent is time-consistent and $\beta=1$. Furthermore, for the normal distribution the squared difference in stopping probabilities in the sophisticated and naive case are remarkably small (less than $0.232\%$), so (in  a finite data set analogue) nothing would indicate to the analyst that these are bad distributional choices to model the unobservable shocks.\footnote{ If the analyst does not know the mean and standard deviation of the shock distribution and thus needs to estimate these parameters as well, she is able to fit the data even better, making it even harder to detect her misspecification.}
\end{example}
\medskip

Our general results below, which establish that non-parametrically the degree of time-inconsistency is never identified from task completion data, prove that the above examples are not artefacts of the numbers we have chosen. For every set of model parameters $\delta,\beta,\subbeta$ and any given dataset, there exists some unobserved stationary payoff distribution that perfectly fits the data. Thus, the analyst can rule out parameter values for $\delta,\beta,\subbeta$ only through ad-hoc assuming a specific parametric family of distributions. As a consequence, the analyst's conclusions are---in  line with Example 2---solely determined by her parametric choice for the unobservable payoff distribution.

\section{Preliminary Analysis: Recursive Structure}
\label{sec:recursive}

We begin by establishing that the agent's problem has a simple recursive structure. A strategy $\sigma_t(\cdot,\cdot; z)$ is a cutoff strategy with cutoffs $z=(z_1,\ldots,z_T)$ if
\[
	\sigma_t(Y^{t-1},y_t;z) = \begin{cases}
		0 & \text{ if } y_t \leq z_t \\
		1 & \text{ if } y_r > z_t 
	\end{cases}	\,.
\]
Self $T$ completes the task if and only if her realized payoff is strictly greater than $\by$. Furthermore, selves $t<T$ believe that Self $T$ will complete the task if and only if her realized payoff is strictly greater than $(\nicefrac{\subbeta}{\beta}) \by$. Hence both the perceived and actual strategy in the final period are cutoff strategies. Similarly, if all future selves are perceived to use cutoff strategies, Self $t$ can calculate the perceived continuation value of waiting, and will complete the task if and only if her current payoff is greater than this perceived continuation value. Hence, by induction, all selves use a cutoff strategy and perceive their future selves to use a cutoff strategy. 

For a partially-naive quasi-hyperbolic discounter, the time $t$ and time $t'$ selves have the same beliefs about the strategy future selves---i.e. selves active after time $\max\{t,t'\}$---use. 
Self $t$ thus believes that if she does not complete the task at time $t$, the task will be completed at the (random) time
\[
	\subtau_t = \min\{s>t\colon y_s > c_s \} \,,
\]
where $c_s$ is the perceived cutoff that selves $t<s$ believe Self $s$ will use. Trivially, for all $s>t$ the stopping time $\subtau_{s}$ equals $\subtau_{t}$ conditional on not stopping before time $s+1$,
\[
	\PP[\subtau_s = \subtau_t \mid \subtau_t > s] = 1\,.
\]
Hence, Self $t$  believes that her \textsl{perceived continuation utility} $v_t$ if she does not complete the task at time $t$ is given by
\[
	v_t = \beta \, \EE \left[ \delta^{\subtau_t-t} y_{\subtau_t} \right] \,.
\]
Since Self $t$  stops whenever the value of completing the task immediately is greater than her subjective continuation value, the time $\tau_t$ at which the task is completed conditional on not having been completed before time $t$ is given by
\[
	\tau_t = \min \{s>t \colon y_s > v_s \} \,. 
\]
We first show that the perceived continuation values satisfy a recursive equation.\footnote{Throughout this paper, $\int \cdot \, dF$ denotes the Riemann--Stieltjes integral.} 
\begin{lemma}[Recursive Characterization]\label{lem:rec-representation}
A pair of strategies $(\sigma,\hat{\sigma})$ constitute a perception-perfect equilibrium if and only if both are cut-off strategies with cutoffs $(v,c) \in \RR^T \times \RR^T$ that satisfy the equations
\begin{equation}\label{eq:rec-representation}
	v_t = \begin{cases} \beta \,\delta \int_{\nicefrac{\subbeta}{\beta}\,\,v_{t+1}}^\infty z \,d F(z) + F(\nicefrac{\subbeta}{\beta}\,\,v_{t+1}) \, \delta \, v_{t+1} & \text{ for } t<T\\
	\by &\text{ for } t=T\end{cases}
\end{equation}
and $c_t = \left(\nicefrac{\subbeta}{\beta}\right) \,\,v_t$.
\end{lemma}

\begin{proof}
We first show that the conditions are necessary for a perception-perfect equilibrium.
We already argued that any equilibrium must be in cutoff strategies and that the cutoffs used by each self must equal their perceived continuation value $v$.
We can rewrite the perceived continuation utility by considering the event that the task is completed in period $t+1$ as well as the complementary event that it is completed later
\begin{align*}
	v_t &= \beta \, \, \EE \left[ \delta^{\subtau_t-t} y_{\subtau_t} \right] 	= \beta \, \, \EE \left[ \mathbf{1}_{\subtau_t = t+1} \delta^{\subtau_t-t} y_{\subtau_t} + \mathbf{1}_{\subtau_t > t+1} \delta^{\subtau_t-t} y_{\subtau_t} \right] .\\
	\intertext{Because Self $t$ believes the task is completed in period $t+1$ if and only if the benefit is greater than the subjective cutoff $y_{t+1} > c_{t+1}$, this equals}
	v_t &= \beta \, \, \EE \left[ \mathbf{1}_{y_{t+1}> c_{t+1}} \delta y_{\subtau_t} + \mathbf{1}_{y_{t+1}\leq c_{t+1}} \delta^{\subtau_t-t} y_{\subtau_t} \right] \,.
\end{align*}
Since $y_{t+1}$ is distributed according to $F$ and $\subtau_t = \subtau_{t+1}$ conditional on not stopping in period $t+1$, we can use the definition of a Riemann--Stieltjes integral to rewrite the above as
\begin{align*}
	v_t &= \beta \delta \int_{c_{t+1}}^\infty z \,d F(z) + F(c_{t+1}) \,  \beta  \delta \, \EE \left[ \delta^{\subtau_{t+1}-(t+1)} y_{\subtau_{t+1}} \right].
\end{align*}
Using the definition of $v_{t+1}$ to rewrite the last summand above, we therefore have that
\begin{equation}\label{eq:v-c-dynamic}
v_t = \beta \delta \int_{c_{t+1}}^\infty z \,d F(z) + F(c_{t+1}) \, \delta \, v_{t+1}\,.
\end{equation}
Here, $v_t$ is the cutoff that Self $t$ actually uses. Prior selves, however, believe that Self $t$ discounts with hyperbolic weight $\subbeta$, so the perceived cutoff $c_t$ they think Self $t$ uses solves
\begin{align*}
	c_{t} &= \subbeta \, \EE \left[ \delta^{\subtau_t-t} y_{\subtau_t} \right] =  \left(\nicefrac{\subbeta}{\beta}\right)\, \beta\, \EE \left[ \delta^{\subtau_t-t} y_{\subtau_t} \right] = \left(\nicefrac{\subbeta}{\beta}\right) \,\,v_t \,. 
\end{align*}

\noindent Using this equation to replace $c_{t+1}$ in \eqref{eq:v-c-dynamic} establishes that the continuation values $v_1,\ldots,v_{T-1}$  satisfy the recursive equation
\[
	v_t = \beta \, \delta \int_{ \left(\nicefrac{\subbeta}{\beta}\right)\,\,v_{t+1}}^\infty z \,d F(z) + F( \left(\nicefrac{\subbeta}{\beta}\right)\,\,v_{t+1}) \, \delta \, v_{t+1}\,.  
\]
That any such pair of cutoff strategies constitutes a perception-perfect equilibrium follows from checking the (perceived) optimality conditions inductively starting from the last period.
\end{proof}%
To see the intuition behind Equation \ref{eq:rec-representation}, suppose first that the agent is sophisticated ($\subbeta=\beta$) in which case ($\nicefrac{\subbeta}{\beta}=1$). Then the first term is the discounted benefit of stopping tomorrow, which the agent does whenever the benefit of stopping falls above the continuation value of tomorrow's self. This payoff is discounted according to Self $t$'s short-term discount factor $\beta \delta$. The second term captures the fact that with probability $F( v_t)$ tomorrow's self continues because it prefers its perceived continuation value $v_{t+1}$. As today's self discounts payoffs that realize after period $t+1$ by a factor of $\delta$ more than tomorrow's self, this term is discounted with $\delta$. When predicting future behavior, a partially naive agent uses the perceived cutoffs $c_{t} = (\nicefrac{\subbeta}{\beta})\, v_t $ determined by the continuation value a former time $s<t$ self believes Self $t$ has.  If $\subbeta>\beta$, current selves overestimate future selves' patience and, hence, the cutoff they use. If $\subbeta<\beta$,  current selves underestimate future selves' patience and, hence, their cutoffs.





\section{Rate of Task Completion Increases Over Time}
\label{sec:task_completition}

Building on this recursive formulation, this section establishes that a partially-naive quasi-hyperbolic agent is (weakly) more likely to stop and complete the task, the closer she is to the deadline $T$. In other words, the longer away the deadline, the higher the perceived continuation value of the current self. Because the payoff distribution is stationary, comparing the perceived continuation value of period $t$ to that of period $t+1$ is equivalent to comparing the perceived continuation in the first period of task-completion
with a deadline of $T-t$ to that with a deadline of $T-(t+1)$. Interestingly, since the perceived continuation value increases in the distance to the deadline, therefore, a quasi-hyperbolic agent would never want to impose an earlier deadline to keep herself from procrastinating excessively. While obvious for an exponential discounter---adding an extra period simply increases her choice set and hence makes her better off---the question of whether to limit future selves delay possibility is much more subtle when the agent is a quasi-hyperbolic discounter. Indeed, when the distribution of net benefits is not stationary, it is easy to construct counterexamples in which Self 1 would want to impose an early deadline on future selves.

\begin{example}[Self 1 wants to impose a deadline with a time-dependent payoff distribution]\label{ex:deadlines-help} Consider a sophisticated agent with $\delta =1, \beta =1/2$ who has two periods to complete a mandatory task, and who has a deterministic cost of $0.9$ in the first and $1$ in the second period. Due to her present bias, the agent will complete the task in period 2 giving Self 1 a utility of -1/2. Now add the chance to complete the task in a third period at a cost of 1.5. Then Self 2 strictly prefers to procrastinate, and if Self 1 waits, her utility is -3/4. Thus, adding another period in which the task can be completed makes Self 1 worse off. As a result, Self 1 would be willing to impose a two-period deadline.
\end{example}

Intuitively, because preferences between today's self and future selves are not aligned, if payoffs depend on time restricting future selves' choices through imposing a deadline can be beneficial to today's self. \cite{bisinhyndman} provide further examples in which a sophisticated quasi-hyperbolic agent benefits from imposing a deadline when costs of doing a mandatory task follow a Markov process in which higher costs today are associated with higher costs tomorrow.\footnote{While in our simple example the state changes deterministically, continuity of payoffs implies that the example also hold if with a small probability the costs are redrawn from a uniform distribution over $\{0.9,1,1.5\}$ and otherwise move up deterministically towards the state $1.5$ as in our example.}  What is perhaps surprising is that if costs---or net benefit in our setup---are uncorrelated over time, a sophisticated quasi-hyperbolic agent \textsl{never} wants to impose a deadline.

Indeed, when the payoff distribution is the same across periods, we have:

\begin{theorem}[Monotonicity of the Continuation Value]\label{prop:monotone-values} Let $\delta \leq 1$.
\begin{compactenum}[i)]
\item The subjective continuation values are non-increasing over time
\[
	v_1 \geq v_2 \geq \ldots \geq v_T \,.
\]
\item Every self $t$ prefers a later deadline.
\end{compactenum} 
\end{theorem}

Parts $i)$ and $ii)$ are equivalent since when the payoff distribution is identical across periods, the subjective continuation value in a given period $t$ equals the value in the problem with a deadline of $T-t$ periods. To understand intuitively why a quasi-hyperbolic agent's Self $1$ does not want to impose a deadline with a stationary payoff distribution, consider first the case in which doing the task is always costly---i.e., where the support of $F$ is a subset of $\RR_-$. When comparing a $(T-1)$-period to $T$-period deadline, Self 1 realizes that if she does not engage in the task in the $T$ period problem, self 2 will face a $T-1$-period problem. That subgame is identical to the one she faces in the $T-1$ period problem, and future selves who are $s$ periods away from the deadline will behave identically in the two problems. Hence for $s \in {1,\cdots T-1}$, the task completion probability $s$-periods before the deadline is identical, and due to discounting of future costs, Self 1 is strictly better off selecting the $T$-period problem and not doing the task in the first period.

Suppose now instead, that the agent is sophisticated with quasi-hyperbolic parameter $\beta <1$ and that the payoff of completing the task is always positive---i.e., the support of $F$ is a subset of $\RR_+$. From the perspective of a Self $t$, future selves are to impatient, and therefore to willing to cash in the positive benefit in every future period. Suppose now that Self 1 can extend the deadline from $T-1$ to $T$ periods. In this case, Self $T-1$ will wait for sufficiently low net benefits. Because the time $T-1$ self is more impatient than Self 1 would want it to be, whenever the impatient Self $T-1$ chooses to wait, Self 1's expected payoff increases from waiting. Thus, conditional on reaching period $T-1$, the longer deadline benefits Self 1. Now consider Self $T-2$. With the longer deadline, Self $T-2$'s benefit from waiting increases because it always prefers its future self to not complete the task when the future self chooses to do so. Hence, Self $T-2$ will also act less impatiently, which again benefits Self 1 conditional on reaching period $T-2$. By induction, hence, in expectation Self 1 benefits in every future period from the deadline extension.

Because a partially naive Self 1 thinks that she is sophisticated, and in either case a sophisticated agent's Self 1 does not want to impose a deadline, a partially naive agent will not want to do so either. Hence, the perceived continuation value of a partially naive agents also increase in the distance to the deadline.

Our proof studies properties of solutions to the recursive equation \eqref{eq:rec-representation} to extend the above intuitions to cases in which the support of the payoff distribution may contain positive and negative elements, and hence some future selves can be a priori to eager and others not eager enough to complete the task.

We now turn to an immediate implications of Theorem  \ref{prop:monotone-values}.
Note that the probability $p_t = \PP[\tau_{t-1} = t]$ that the agent stops in period $t$ conditional on not having stopped before is the probability that the value of completing the task $y_t$ is above the subjective continuation value $v_t$; i.e.
\[
	p_t = \PP[ y_t \geq v_t ] = 1-F(v_t)\,.
\]
As the subjective continuation value $v_t$ is non-increasing, we have that the objective probability $p_t$ that the agent stops in period $t$ is non-decreasing.

\begin{corollary}\label{cor:monotone-stopping} Let $\delta \leq 1$.  For any given benefit distribution $F$ and in every perception-perfect equilibrium, the objective probability with which the agent completes the task conditional on not having completed it before is non-decreasing towards the deadline, i.e. 
\[
	p_1 \leq p_2 \leq \ldots \leq p_T\,.
\]
\end{corollary}
%


Independently of the naivete  and preference-parameters of a hyperbolic discounter, Corollary \ref{cor:monotone-stopping} provides a simple testable prediction about her task-completion behavior when payoffs are independently and identically distributed over time: the likelihood of completing the task is increasing over time. Section \ref{sec:discussion}, however, emphasizes that researchers need individual not group data to test this prediction.\footnote{Interestingly, this result holds independently of whether the agent over- or underestimates estimates $\beta$, i.e. whether $\subbeta<\beta$ or $\subbeta\geq\beta$.}

\begin{remark}
Corollary \ref{cor:monotone-stopping} establishes that the probabilities of stopping \textsl{conditional} on not having stopped previously increase over time. The unconditional stopping probability, however, may either increase or decrease. This difference is of practical relevance: for example, the conditional stopping probabilities increase over time in the tax-filing data of  \cite{martinezmeier} while the unconditional stopping probabilities decrease.\footnote{See Figure 1 and 2 in \cite{martinezmeier}.}
\end{remark}

\section{Time-Preferences are Unidentifiable from Task Completion}
\label{sec:unidentified}

In this section, we identify a strong sense in which time-preferences are unidentifiable from task completion choices. Recall that we established that for any arbitrary preference profile $\beta,\delta$ and any belief $\hatb$, the profile of stopping probabilities is non-decreasing. 
In this section we establish the converse: absent (parametric) restrictions on the payoff distribution $F$, we show that any non-decreasing profile of stopping probabilities is consistent with any arbitrary preference profile $\beta,\delta$ in case either the agent is  either sophisticated ($ \hatb = \beta$) or fully naive ($\hatb =1$). Hence, it is impossible, for example, to distinguish a naive time-inconsistent agent from a time-consistent one based on their task-completion behavior. Importantly, this impossibility continues to hold even if a researcher is willing to exogenously impose that the ``long-run discount factor'' $\delta$ equals $1$, as is plausible in many applications in which one observes task completion on a frequent (e.g. daily) basis. Similarly, even if the researcher is willing to impose a priori restrictions on plausible levels of $\beta$---including the strong requirement that the agent is time-consistent---, absent exogenous restrictions on $F$, no information on $\delta$ or $\beta$ can be inferred from the task-completion data. 

 Intuitively, whether a self prefers to do a task today or tomorrow depends on her time preferences (as well as beliefs about future selves' time preferences) and on the perceived option value of waiting. The option value of waiting, in turn, depends on the payoff distribution. Through changing the unobservable payoff distribution, we can hence undo a change in the present-bias or long-run-discount factor of the agent. Technically, however, a local change in the payoff distribution affects continuation values in every period in a highly non-linear way, so to establish that we can construct an appropriate payoff distribution, we need a non-local argument. When the agent is either sophisticated or fully naive---for different technical reasons that we explain below---the analysis simplifies and allows us to establish that we can indeed rationalize the stopping behavior for any arbitrarily chosen $\beta,\delta$.
 
For the case in which the penalty is unobservable, we furthermore illustrate that the data is rationalizable as the optimal behavior of a fully patient time-consistent agent $(\hatb=\beta=\delta=1)$ facing an unobservable payoff distribution $F$ with {\it any given} expected value and (non-zero) variance of the distribution; any parametric identification of present bias in such a task-completion setting, therefore, must be based on a prior knowledge of higher-order moments of the benefit distribution.

\subsection{Time-Preferences are Unidentifiable: Sophisticated Case}

In this subsection, we establish that absent (parametric) restrictions on the payoff distribution $F$, any non-decreasing profile of stopping probabilities is consistent with any arbitrary preference profile $\beta,\delta$ of a sophistcated quasi-hyperbolic discounter. In particular, we have: 

\begin{theorem}[Non-identifiability]\label{thm:non-identifiability-sophisticate}
Suppose the agent is sophisticated $\subbeta=\beta$. For every non-decreasing sequence of stopping probabilities $0 < p_1 \leq p_2 \leq \ldots \leq p_T < 1$, every $(\delta,\beta) \in (0,1] \times (0,1]$, and every penalty $\nicefrac{\by}{\beta \delta} \in \RR$, there exists a distribution $F$ that rationalizes the agent's stopping probabilities as the (unique) outcome of a  perception perfect equilibrium.
\end{theorem}

Technically, to prove the theorem, we construct a distribution with $t+2$ mass points, where each of the non-extreme values equals the agent's (correctly perceived) continuation value in a given period $t \in \{1,\ldots,T\}$; i.e. the second lowest mass point is set at the value $v_T = \by$, and so on. The probability on each mass point is chosen so that the agent---who waits if and only if $y_t \geq v_t$---selects the exogenously given stopping probability. The constructions is feasible since when $\hatb =\beta$, the recursive representation (Lemma \ref{eq:rec-representation}) takes a particular simple form, and together with the chosen construction of the distribution gives rise to a system of linear equations, which can be solved forward.


\subsection{Time-Preferences are Unidentifiable: Naive Case}

We now turn to the case  in which the agent believes to be time-consistent and establish that for every chosen non-decreasing sequence of stopping probabilities and every chosen preference profile $\beta,\delta$, there exists a payoff distribution $F$ that admits a piecewise constant density and induces the agent to choose the stopping behavior given by the data.

\begin{theorem}[Non-identifiability]\label{thm:non-identifiability}
Suppose the agent believes to be time-consistent $\subbeta=1$. For every non-decreasing sequence of stopping probabilities $0 < p_1 \leq p_2 \leq \ldots \leq p_T < 1$, every $(\delta,\beta) \in (0,1) \times (0,1]$, and every penalty $\nicefrac{\by}{\beta \delta}<0$, there exists a distribution $F$ that rationalizes the agent's stopping probabilities as the unique outcome of any  perception perfect equilibrium.
\end{theorem}


Our formal proof in the appendix proceeds roughly as follows. Step (i). Fix the agent's time preference as well as period $T$'s continuation value (which equals $\by$). Step (ii). Take an arbitrary $(T-1)$-element vector of  non-increasing continuation values $v_1 \geq v_2 \geq \ldots \geq v_{T-1}$. Step (iii). Here, we generate a payoff distribution for these continuation values that gives the desired stopping probabilities. In particular, we put a probability mass that is equal to the difference in the exogenously given stopping probability between period $t$ and $t+1$ between the corresponding period's perceived continuation values, for simplicity using a uniform density. This step, hence, amounts to mapping continuation values into distributions that lead to the correct stopping probabilities.  
Step (iv). Calculate the actual continuation values that the new payoff distribution from the third step gives rise to. This maps the set of distributions back into the vector of continuation values. By Theorem \ref{prop:monotone-values}, these continuation values are again non-decreasing, and thus the combined function maps a non-increasing sequence of continuation values into a non-increasing sequence of continuation values. Step (v). We show that this function is bounded and maps sequences from an appropriately chosen interval into itself. Furthermore, the function is monotone as higher continuation values lead to a better distribution (in the sense of first-order stochastic dominance) and a better distribution increases the subjective continuation values for an agent who believes to be time-consistent (established in Lemma \ref{lem:aux-properties-naive} $ii)$ below). Thus, the mapping from continuation values into continuation values is a monotone mapping from a complete lattice into a complete lattice, and by Tarski's Theorem admits at least one fixed point.
Any fixed point gives the desired distribution, since by Step (iii) the stopping probabilities are correct and by Step (iv) the continuation values are those consistent with the limit distribution. Furthermore, because by  Lemma \ref{lem:aux-properties-naive} $i)$ below, the continuation values are strictly decreasing when $F(\by)>0$ and $\by<0$, the limit distribution that we construct is continuous, so that the agent's stopping behavior is unique.


As explained in the above sketch, the proof of Theorem \ref{thm:non-identifiability} relies on the following Lemma.

\begin{lemma}\label{lem:aux-properties-naive}
Suppose $\delta<1$ and the agent believes to be time-consistent $\subbeta=1$.
\begin{compactenum}[i)]
\item For every distribution $F$ with $F(\by)>0$ and $\by<0$, the continuation values are strictly decreasing $v_1 > v_2 > \ldots > v_T$.
\item {A first-order stochastic dominance increase in the payoff distribution $F$ increases the vector of subjective continuation values point-wise.}
\end{compactenum}
\end{lemma}

Part $i)$ shows that whenever there is a positive probability that the utility from completing the task in the final period before the deadline $\by$ is less than that from not completing the task, an agent who believes to be time-consistent (i.e. who has beliefs $\hatb=1$) has a {\it strictly} positive willingness to pay for extending the deadline. Here, the assumption that $F(\by)>0$ and $\by<0$ rules out that it is optimal for the agent to always complete the task immediately.\footnote{As a trivial counterexample to the finding when the assumption is dropped, suppose the task yields a (net) positive deterministic payoff above $\by$. Then the agent would always complete the task immediately, and hence is unwilling to pay for extending the deadline.} Thereby, it allows us to strengthen the finding of Theorem \ref{prop:monotone-values} for the case of $\subbeta=1$.
 
The second part of the Lemma shows that any improvement in the payoff distribution weakly increases the subjective continuation values in all periods. Obviously, for a time-consistent agent an improvement in the payoff distribution raises the second to last period's continuation payoff. Furthermore, from the third to last period's perspective, the increase in the payoff distribution and the penultimate period's continuation value, makes it more desirable to reach the second to last period, that is increases its continuation value; etc... . And because an agent with beliefs $\hatb=1$ thinks she is time-consistent from tomorrow on, it similarly increases her continuation values.

While economically we do not believe that the restriction to fully naive or actually time-consistent agents (with $\hatb =1$) is important for Theorem \ref{thm:non-identifiability} to hold, our mathematical proof uses this assumption when arguing that subjective continuation values increase in a first-order-stochastic dominance shift in the payoff distribution, which in turn allows us to use Tarski's Theorem. In general, due to a time-inconsistent agent's the conflict of interest between her different selves, a first-oder-dominance improvement of her payoffs need not raise subjective continuation values as the following example highlights.

\begin{example}[A sophisticated $\beta,\delta$-agent can prefer a fixed uniformly payoff-reducing tax]\label{ex:tax-sophisticate}

 Let $\beta = 1/8$ and the agent be sophisticated ($\subbeta=\beta$). To simplify the calculation, we set $\delta =1$ but the argument obviously extends to $\delta$ sufficiently close to $1$.  We compare the agent's expected welfare and (subjective) continuation values in a three-period voluntary-task-completion problem across two scenarios.\footnote{Because even the lowest payoff from completing the task is positive, the agent always completes the task voluntarily. Our results, thus, remain unchanged if task completion becomes mandatory.} One without a tax, and one in which the agent has to pay a fixed utility tax of $1/8$ in the period in which she completes the task. Let the distribution $F$  of payoffs absent a tax be such that with probability $3/4$ the agent receives a payoff of $3/2$, and with the remaining probability of $1/4$ the agent receives a payoff of $1/4$. Straightforward calculations (see the Supplementary Appendix) establish that the agent strictly prefers the tax to the no tax situations and that the tax increases the first-period continuation value.
 \end{example}
 
Note that the tax introduced in Example \ref{ex:tax-sophisticate} is the same independent of when the agent completes the task and in that sense is not tailored to punish an agent for giving in to early temptations. Intuitively, nevertheless, the tax in the above example lowers the temptation to quit immediately in period 2 as it reduces the benefits from doing so. As a result, the agent obtains a commitment device to only stop when payoff are high in the second or first period. The benefits thereof overcompensate the direct payoff reduction through the tax, and thereby raise earlier periods' continuation values. 

Lemma  \ref{lem:aux-properties-naive} and Example \ref{ex:tax-sophisticate} jointly imply that one can (sometimes) identify agents that believe to have self control problems $(\hatb<1)$: such an agent can have a strictly positive willingness to pay to make his payoff distribution strictly worse. In contrast, an agent who believes to be time-consistent ($\hatb =1$) and hence does not foresee future self-control problems will never want to do so. 


\subsection{Known Expected Value and Variance}
\label{subsec:moments}

For our very general results, we have not restricted the class of permissible distribution functions. One may hope to rule out time-consistency and find evidence through restricting features of the distribution. Perhaps the most natural way of doing so would be two make restrictions regarding the moments of $F$; for example, an analyst may have an idea regarding the possible expected net benefit of doing the task---that is regarding the mean of $F$---or may be willing to impose that net benefits do not vary to much between periods (restricting the variance of $F$).

We now briefly observe that if the penalty is unobservable, even with a priori knowledge of the mean and variance of $F$ it is impossible to rule-out time-consistent behavior. To see this, consider an agent for whom $\beta =\delta =1$. Theorem \ref{thm:non-identifiability-sophisticate} implies that there exists a net benefit distribution $F$ that rationalizes any increasing profile of stopping probabilities. Furthermore, in this case the recursive formulation of the problem in Lemma \ref{eq:rec-representation} simplifies to
$$
v_t = \EE \left[ \max \{ y_{t+1},v_{t+1} \} \right] \ \ \ \ \text{ for all } t <T.
$$
Hence, if the distribution $F$ together with the penalty $\by$ rationalize the data, so does the distribution $F + \kappa$ together with the penalty $\by +\kappa$ for any $\kappa \in \RR$. In other words, we can always select a net benefit distribution with a given expected value. Furthermore for any $\kappa_2 >0$, the stopping behavior remains optimal if we scale the net-benefits and $\by$ by $\kappa_2$. This implies that we can not only select a distribution with a given mean but that we can at the same time select any desired variance and explain the observed stopping behavior.\footnote{Indeed, since the construction of $F$ in the proof of Theorem  \ref{thm:non-identifiability-sophisticate} uses bounded support, we can rationalize the observed stopping behavior as resulting from a patient agent ($\beta =\delta =1$) whose net benefits vary arbitrarily little.}

\begin{corollary} 
Suppose the agent is time-consistent and fully patient $\subbeta=\beta=\delta=1$. For every non-decreasing sequence of stopping probabilities $0 < p_1 \leq p_2 \leq \ldots \leq p_T < 1$, and every $\mu \in \RR$ and  $\sigma^2>0$, there exists a distribution $F$ with mean $\mu$ and variance $\sigma^2$ and a penalty $\by$ that rationalizes the agent's stopping probabilities as the (unique) outcome of a  perception perfect equilibrium.
\end{corollary}

\section{Non-Parametric Identification with Richer Data}

\label{sec:rich data}

Above, we established that stopping data by itself is insufficient to test for time preferences. A natural question is whether richer data allows the analyst to learn about the agent's time-preferences. To do so, the analyst needs to disentangle whether the stopping behavior is driven by a desire to delay incurring costs or by the option value of drawing a better payoff in the future. Observe that in the latter case, a considerable option value requires payoff to differ significantly. Hence, as the deadline approaches and a waiting agent faces fewer future draws, the continuation value should drop considerably. In contrast, even with a (relatively) constant option value, an agent who is present biased is willing to delay a costly activity to the last minute. Thus, observing, in addition to task-completion times,  continuation values directly should facilitate the non-parametric identification of $\delta,\beta,\subbeta$. We, thus, analyze how much the analyst can learn when also observing the continuation values.

More formally, consider the case in which the analyst observes the agent's stopping behavior (infinitely often) as well as his exact willingness to pay for continuing with the task. Conceptually, the analyst could elicit this information by selecting some stopping problems in which she offers the agent a mechanism at the end of period $t$ that truthfully elicits her willingness to pay for continuing with the task from $t+1$ onwards.\footnote{If the analyst sees infinitely many identical agents, she can randomly select $T$ agents. Label these agents $k=1, \ldots, T$. At the end of period $k$, the analyst then elicits agent $k$'s willingness to pay for facing the task-completion problem from period $k+1$ to $T$. She can do so using a standard  Becker-De Groot-Marschak mechanism \citep{beckerdegroot}. } Denote the amount she is willing to pay at the end of period $t$ by $m_t$. If the agent's utility is quasi-linear in money, which is a good approximation in the standard hyperbolic discounting model whenever the involved stakes are relatively small---as in the case of parking tickets---, then observing $m_t$ is equivalent to observing the continuation value $v_t$; otherwise, $v_t = u(m_t)$ for some monotonically increasing utility function $m_t$. We provide an exact analytical result regarding partial identification for the case of linear utility in money and a sophisticated agent. But, we also highlight that---at the cost of having to use numerical methods common in empirical work to solve for the admissible parameter range---our results can be readily extended in multiple directions, including partial naivete and non-linear utility in money. Importantly, below we also point out that our procedure identifies the time-preferences over effort even if the agent discounts money---due to time-preferences or the ability to borrow or save---differently than effort, which implies that our time-preference identification is robust to standard criticisms of eliciting time preferences using monetary choices \citep{augenblickniederle,ericsonlaibsonreview,ramsey}.

As a preliminary observation, recall that Theorem \ref{prop:monotone-values} and Corollary \ref{cor:monotone-stopping} imply that the elicited continuation values must be non-increasing and the observed stopping probabilities non-decreasing. We refer to data $v,p$ that has these properties as {\it plausible}.\footnote{If $\bar{y}$ is observable then in addition we require that $v_T = \bar{y}$.} Any data that is not plausible cannot be justified by our quasi-hyperbolic setup. 
Imposing that the agent is sophisticated, we now show how to non-parametrically identify the set of $\beta,\delta$ that are consistent with the observed data.  Using Lemma \ref{lem:rec-representation} and the fact that an agent stops whenever his payoff is strictly above the continuation value, for a sophisticate the continuation values $v$ and conditional stopping probabilities $p$ must satisfy
\begin{align}\label{eq:constraints-sophisticate}
\begin{aligned}
	v_t &=u(m_t)  \ \  &&\text{ for all } t \in \{ 1,\ldots, T\}\, ,\\
	\int_{v_{t+1}}^\infty z \,d F(z)  &=  \frac{\delta^{-1}\,v_t - (1-p_{t+1})  \, v_{t+1}}{\beta} \ \ \  &&\text{ for all } t \in \{ 1,\ldots, T-1\}\, ,\\
	1 -  F(v_t) &= p_t \ \  &&\text{ for all } t \in \{ 1,\ldots, T\}\,.
\end{aligned}
\end{align}
Conversely, if a pair $u,F$ satisfies \eqref{eq:constraints-sophisticate} for a given plausible data set, then  Lemma \ref{lem:rec-representation} implies that it gives rise to a perception perfect equilibrium for a sophisticated agent.

Note that the right-hand-side of \eqref{eq:constraints-sophisticate}  is given by the data and hypothesized values of $\beta$ and $\delta$. Thus, the data is consistent with a given pair $\beta,\delta$  if and only if there exists a distribution $F$ that solves \eqref{eq:constraints-sophisticate}. As a preliminary step, we show that whenever \eqref{eq:constraints-sophisticate} admits a solution, it also admits a solution that is a distribution consisting of $T+1$ mass points.  

\begin{lemma}\label{lem:mass_points_sufficient}
Whenever \eqref{eq:constraints-sophisticate} admits a solution for a plausible data set, there exists a solution $F$ that consists of exactly $T+1$ mass points located at $(\pi_0, \ldots,\pi_T)$ that satisfy
$$
\pi_0 \leq v_T < \pi_1 \leq v_{T-1} <  \ldots \leq \pi_{T-1} \leq v_1 < \pi_T,
$$ 
with associated probabilities $f_k=\PP[y = \pi_k] $ given by
\begin{equation*}
	f_k = \begin{cases} 1-p_T &\text{ if } k = 0\\
p_{T-k+1} - p_{T-k} &\text{ if } k \in \{ 1,\ldots,T-1\}\\
p_1  &\text{ if } k = T
\end{cases}\,.\
\end{equation*}

\end{lemma}

%

Intuitively, two distributions give rise to the same stopping probability when the probability mass above the continuation values is the same. And the only things that matters for the option value of waiting is the probability with which the agent stops at given future point in time and the expected payoff conditional on doing so. By moving the probability mass between any two continuation values to the expected payoff conditional on falling between these two values, thus, the incentives to wait are unaltered. Furthermore, because the observed stopping probabilities determine the continuation mass between any two continuation values, the question of whether the analyst can non-parametrically match the observed data for a given $\beta, \delta$ boils down to the question of whether she can do so by choosing a distribution consisting of $T+1$ mass points in the appropriate intervals. 

Conceptually, Lemma \ref{lem:mass_points_sufficient} hence allows the analyst to search over a finite dimensional rather than an infinite-dimensional space of possible distribution. Indeed, under the distributional restriction given by the lemma, \eqref{eq:constraints-sophisticate} becomes a non-linear system with finitely many real-valued unknowns. Theorem \ref{thm:non-para_identification}, which we prove in the Appendix, shows that this system can be transformed into a simple set of transparent inequalities that identify the values of $\delta$ and $\beta$ that are consistent with the observed stopping behavior and elicited continuation values.

\begin{theorem}[Non-Parametric Identification] \label{thm:non-para_identification} Suppose $u(m_t) =m_t$ for all $t$ and that $p_1 >0$.\footnote{We require $p_1>0 $ only to simplify the statement. \comments{check carefully}} Plausible data $(v,p)$ is consistent with $\beta,\delta$ and sophistication $\subbeta=\beta$ if and only if (i)
\begin{equation}\label{eq:beta-at}
	\beta < \frac{\delta^{-1}\,v_1 - (1-p_{2})  \, v_{2}}{v_{2} (p_2 - p_1) + v_1 p_1} \, \nonumber
\end{equation}
and (ii)  $ v_{t+1} \beta  <  v_{t+1} a(\delta,t) \leq v_t \beta$ for all $t \in \{2, \ldots, T-1\}$, where 
$$
a(\delta,t) =   1 - \frac{\delta^{-1} (v_{t-1} -v_t) - (1-p_t) (v_t - v_{t+1})  }{ v_{t +1}(p_{t+1} - p_t)}.
$$
\end{theorem}

The theorem provides an exact characterization of what time-preference parameters are consistent with the observed rich data. To illustrate its implications, consider the example from Section \ref{sec:example} in which $T=5$, the agent's payoff of completing the task are uniformly distribiuted over $[-1,1]$, and the agent is time-consistent with $\beta=\delta=1$ (this is the setup of Example 2). We illustrate the set of parameters the analyst can identify non-parametrically for $T=5$ and $T=20$ in Figure \ref{fig:example_np}. It is immediate that---in contrast to the case of unobservable continuation values---not all parameter combinations $\beta, \delta$ are consistent with the data.

Figure \ref{fig:example_np}, however, also illustrates that even if the analyst correctly imposes that $\delta =1$, she cannot make precise inference in the case where $T=5$. Indeed, in the example any $\beta$ between $0.82$ and $1.28$ is consistent with the data.
 This changes drastically for $T=20$ in which case $\beta$ is tightly identified once $\delta =1$ is imposed. Without imposing $\delta=1$, however,  the inference about $\beta$ remains imprecise even in the case of $T=20$, as it is impossible to reject $\beta=0.84$. Overall,  the example suggests that rich data---including a significant number of continuation values---are needed for tight parameter estimates.

What allows the analyst to separate the option-value-from-waiting based reason for delaying the task from time-preference-based ones with a rich enough data set? If the agent is patient, he will only delay completing the task with high probability in case he expects a better draw with high probability. This implies that there needs to be considerable variation in the underlying payoff distribution. But then as the deadline moves closer, the agent foresees getting less and less draws, which means the option values quickly drops. In contrast, if time preferences are the underlying reason for delaying, the continuation value will drop much more slowly as the deadline approaches. The additional data on continuation values, hence, allows for set identification of the preference parameters. 
\begin{figure}
	\begin{center}
	  \includegraphics[width=0.7\linewidth]{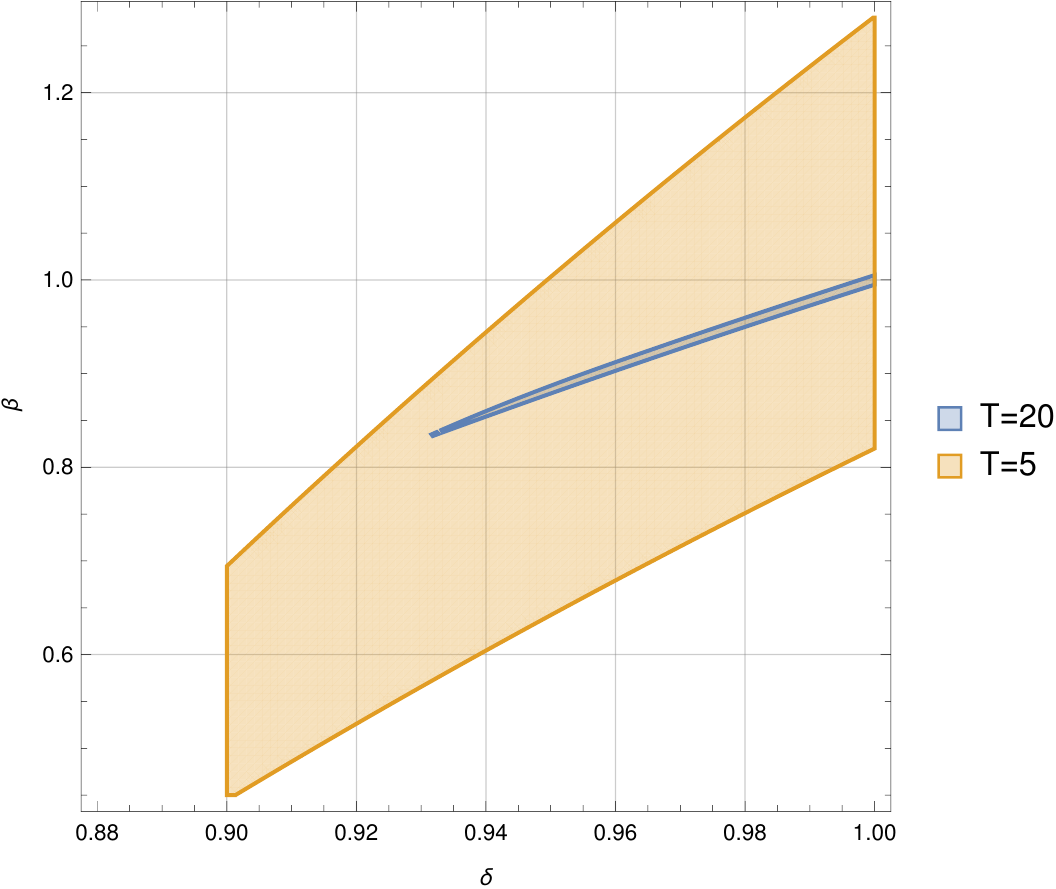}
		\caption{\label{fig:example_np}The above figures illustrates the set of parameters $\beta,\delta$ that the analyst can non-parametrically identify if  she correctly imposes that the agent's instantaneous utility is linear in money. The agent's true values of completing the task are uniformly drawn from $[-1,1]$ and she is time-consistent with $\beta=\delta=1$. In yellow is the case of $T=5$ periods of data and in blue $T=20$ periods of data.}
	\end{center}
\end{figure}
Since it is the change in option value that allows identification, one can also use other related data. For example, the willingness to pay for extending the deadline reflects the drop in continuation value, and therefore would also give rise to a rich data set that would allow non-parametric set identification. Again, however, our example suggests that many such observations are needed, suggesting that a tight estimation of agents' time inconsistency requires ``extremely rich'' task-completion data.

\paragraph{Generalizations of this Methodolgy} We think of the Theorem \ref{thm:non-para_identification}  as a proof of concept, and analysts can adopt it to the data  at hand and the assumption they are willing to make. For example, it is in principle straightforward to adopt the above analysis to allow for partial naivete. In that case, however, one needs to be careful to account not only for the probability mass and expectation of falling between two actual continuation values but also differentiate whether a given probability mass falls above or below the anticipated continuation values $c_t$. An analog to Lemma \ref{lem:mass_points_sufficient} implies that this can be done with $2T+1$ mass points. In this case, however, for intervals that are bounded by anticipated and not actual continuation values, the probability that $y_t$ falls into this interval is unknown. As a result, the analyst needs to choose both the mass point and the weight on it (with the appropriate constraints from the observed stopping behavior), giving rise to quadratic constraints. While this can be solved numerically using standard techniques, a simple transparent closed-form solution as in the case of Theorem \ref{thm:non-para_identification} is unavailable. Similarly, because we only need to consider a finite number of mass points, one can allow for non-linear utility in money, which---imposing that utility is increasing in money---requires the analyst to choose increasing utility values $u(m_t)$ in addition to the mass points.\footnote{If the analyst wants to impose risk-aversion in money, this adds simple (linear)  constraints that ensure that the slope of $u$ is non-increasing in $m_t$. Again, this can be solved using standard numerical techniques.}

\paragraph{Time-Preferences over Money} One important aspect of our procedure is that it does not (explicitly or implicitly) impose constraints on how the agent handles monetary payments at different points in time. It is sufficient for contemporaneous utility to be separable in money, and the marginal utility of receiving money to be the same across periods. This assumption is consistent with an intertemporal set-up in which the agent can borrow and lend at given interest rates---in which case the interest rate determines how she trades off monetary payments at different points in time \citep{ericsonlaibsonreview,ramsey}. But it is also consistent with an agent narrow bracketing and consuming small monetary payments immediately---or reasoning as if she does so---so long as she trades of money and effort consistently over time. The procedure outlined in this section thus works for either specification of the agent's time preferences over monetary payments.

\section{Discussion}
\label{sec:discussion}

Our results establish a strong form of non-identifiability in that---absent data on continuation values---even with ideal stopping data in which the analyst observes the exact stopping probability for each individual separately, without parametric assumptions nothing can be learned regarding the agent's discount factor, taste for immediate gratification, or degree of sophistication. In reality, an analyst is likely to observe a large group of agents and infer their average stopping probability; if the group is homogenous our analysis applies. If individuals, however, in addition differ in their unobservable payoff distribution or time preferences, the analyst's problem becomes even more difficult. In that case, for example, it is easy to generate non-monotone stopping probabilities for the overall population. As a simple example, suppose there are two types of agents in the population that face a three-period mandatory task-completion problem. The first type stops in each period with probability 1, while the second type only stops in the final period. If $\alpha >0$ is the fraction of the first type, then the aggregate stopping probability is $\alpha$ in the first period, $0$ in the second, and $1- \alpha$ in the final period, which is clearly non-monotone.\footnote{See \cite{heffetzodonoghue} for a more detailed discussion of heterogeneity as well as empirical evidence on its importance in determining when individuals pay their parking fines.}

Importantly, we establish our formal result for the specific task-completion setting analyzed, and they should not be misconstrued as implying complete non-identifiability of the quasi-hyperbolic discounting model in other settings. In richer and different datasets, it is possible to identify $\beta,\hat{\beta}$ more directly. For example, lotteries (or contracts) that payoff differently depending on the agent's own future behavior can be used to reveal whether the agent missperceives her own future behavior and, hence, whether she is (partially) naive in the quasi-hyperbolic discounting model \citep[see, for example,][]{dellavignamalmendier,Spiegler_2011_book}. Similarly, if the agent is willing to pay for reducing her choice set or for imposing a fine for certain future actions, she values commitment and---within the quasi-hyperbolic discounting framework---must be time-inconsistent \citep[see, for example,][]{strotz}. Such identification strategies, however, rely on data that is fundamentally different from the task-completion data for which we establish the impossibility of non-parametric identification. 

Indeed, even in the closely related, but different, problem of task-timing \citep{carrollchoi,laibson3} in which the benefit from doing the task start accumulating as soon as the agent finishes it, it is possible to construct examples in which an agent wants to commit to an earlier deadline, implying that at least partial identification of perceived present-bias ($\hatb \neq 1$) is theoretically feasible. While agents may theoretically benefit from imposing a deadline in such task-timing problems, however, the calibration of the example in \cite{laibson} suggests that their willingness to do so is small, suggesting that identifying time-inconsistency may nevertheless be challenging in real-world data.

The broader economic lesson from our analysis is that conclusions about time-preferences can quickly be driven by seemingly innocuous parametric assumptions. Our results on set-identification with richer data illustrate, however, that it is possible---and in our setting surprisingly easy---to avoid functional form assumptions. We, thus, think of these results as a proof of concept for the feasibility of non-parametric analysis within the quasi-hyperbolic discounting framework.

Finally, let us emphasize the obvious fact: even though present-bias is non-identifiable in our task-completion settings absent data on continuation values, present-bias may still be a major driver for the wide-spread observation that agents complete tasks last minute. Our results simply caution that the observed task-completion behavior in these settings on its own  is not enough to conclude that present-bias is widespread.

\section*{Appendix}

Define the function $g:\RR \to \RR$ as
\begin{equation}\label{eq:def-g}
	g(w) = \subbeta \, \delta \int_{w}^\infty z \,d F(z) + F(w) \, \delta \, w\,.
\end{equation}
As the following lemma formally establishes, $g$ has a number of convenient properties.
%
%
\begin{lemma}\label{lem:prop-g} The function $g$ has the following properties:
\begin{compactenum}[i)]
\item For all $t \in \{1,\ldots,T-1\}$, the perceived continuation values satisfy $ \left(\nicefrac{\subbeta}{\beta}\right)\,\,v_{t} = g\left( \left(\nicefrac{\subbeta}{\beta}\right)\,\,v_{t+1} \right)$.
\item $g(w)$ is non-decreasing for $w \geq 0$, is right-continuous, and has only upward jumps.
\end{compactenum}
Let $\delta<1$. Then $g$ has the following additional properties:
\begin{compactenum}
\item[iii)] $g(w) > w$ for all $w < 0$ and there exists $\bar{w}>0$ such that $g(w)< w$ for all $w>\bar{w}$. 
\item[iv)] Let $w^\star = \inf \{ w \in \RR \colon g(w) \leq w\}$. Then $w^\star$ satisfies $g(w^\star)=w^\star$ and $w^\star \geq 0$.
\item[v)] If $w' \geq 0 > w$,  then $g(w') \geq g(w).$
\end{compactenum}
\end{lemma}

\noindent {\bf Proof of Lemma \ref{lem:prop-g}:} $i)$ follows immediately from Lemma \ref{lem:rec-representation}. To see that $ii)$ holds, observe that we can rewrite $g$ as 
\begin{align}
	g(w) &= \subbeta\, \delta \int_{w}^\infty z \,d F(z) +  \subbeta F(w) \, \delta \, w + (1-\subbeta) F(w) \, \delta \, w \nonumber \\
	&= \subbeta\, \delta \int_{-\infty}^\infty \max \{ z, w\} \,d F(z) +  (1-\subbeta) F(w) \, \delta \, w \label{eq:alt-representation}\,.
\end{align}
Note that both the first and the second summand are non-decreasing for $w \geq 0$, and that the first summand is continuous in $w$ while the second is right-continuous and has only upward jumps as $F$ is a CDF.

To see that $iii)$ holds, observe that the integral in the first summand of \eqref{eq:alt-representation} is bounded from below by $w$ and, thus, for $w<0$ 
\[
	g(w) \geq \subbeta\, \delta \, w +  (1-\subbeta) F(w) \, \delta \, w = \delta w - (1-\subbeta)(1-F(w))\,\delta\,w \geq \delta w> w \,.
\]
For establishing the second part, note that
\begin{align*}
\lim_{w \nearrow \infty}\frac{g(w)}{w} &=\lim_{w \nearrow \infty}  \left \{ \subbeta\, \delta \int_{-\infty}^\infty \max \left\{ \frac{z}{w}, 1\right\} \,d F(z) +  (1-\subbeta) F(w) \, \delta  \right\} = \subbeta\, \delta +  (1-\subbeta) \, \delta  = \delta\, <1. 
\end{align*}
We next argue that this implies that there exist a $\bar{w}$ such that for all $w> \bar{w}$, $g(w) <w$. Suppose otherwise, then there exists a sequence $w_k \nearrow \infty$ such that $g(w_k) >w_k$. Furthermore, for this sequence $\lim_{w_k \nearrow \infty} \frac{g(w_k)}{w_k} >1$, a contradiction.

We now show $iv)$. Observe that since $g$ has only upward jumps, $w\mapsto g(w)-w$ has only upward jumps. Because by $iii)$  the set $\{ w \in \RR \colon g(w) \leq w\}$ is non-empty, the fact that $w\mapsto g(w)-w$ has only upward jumps implies that $w^\star = \inf \{ w \in \RR \colon g(w) \leq w\}$ satisfies $g(w^\star)=w^\star$.  Furthermore, it follows immediately from $iii)$ that the set $\{ w \in \RR \colon g(w) \leq w\}$ contains only $w \geq 0$, and hence that $w^\star \geq 0$.

To show $v)$, note that for $0 \geq  w$ Equation \ref{eq:alt-representation} together with $w' \geq 0$ implies that
\begin{align}
	g(w') -g(w) = & \subbeta\, \delta \left[ \int_{-\infty}^w (w' -w ) dF(z) + \int_w^{w'} (w' -z) dF(z) \right]  \nonumber \\
	& + (1- \subbeta) \delta \left[ F(w') w' - F(w) w \right] \geq 0,\nonumber
\end{align}
where the inequality follows from the facts that $w' \geq 0$ and $w \leq 0$. \qed
\bigskip

%
%
\noindent {\bf Proof of Theorem \ref{prop:monotone-values}:} That statements i) and ii) of the theorem are equivalent is argued in the main text. Here, we prove statement i).

We begin by establishing the result for $\delta<1$. Trivially, Self $T$'s perceived continuation value is $v_T = y \leq 0$. Define $w^\star = \min \{ w \in \RR \colon g(w) \leq w\}\geq 0$, which is well defined by Lemma \ref{lem:prop-g}, $iii) $ and $iv)$. By Lemma \ref{lem:prop-g}, $i)$ and $iv)$,  we have that
\begin{equation}
	 w^\star - \left(\nicefrac{\subbeta}{\beta}\right) \, v_t = g(w^\star) - g\left(\left(\nicefrac{\subbeta}{\beta}\right) \, v_{t+1} \right) \,.
	 \label{eq_something}
\end{equation}
As $v_T = \by \leq 0$ and $w^\star \geq 0$ (by Lemma \ref{lem:prop-g}, $iv$)), we have that  $\left(\nicefrac{\subbeta}{\beta}\right) \, v_T \leq w^\star$. We now proceed by induction to show that this implies that $v_t\leq w^\star$. 
We distinguish two cases: First,$\left(\nicefrac{\subbeta}{\beta}\right) \, v_{t+1} \geq 0$. In this case the monotonicity of $g$ , established in Lemma \ref{lem:prop-g}, $ii)$, together with Equation \ref{eq_something} implies that $\sgn ( w^\star - \, v_{t} \left(\nicefrac{\subbeta}{\beta}\right) ) = \sgn (w^\star - \left(\nicefrac{\subbeta}{\beta}\right) \, v_{t+1})$ and, thus, by induction $\left(\nicefrac{\subbeta}{\beta}\right) \, v_t \leq w^\star$. Second if $\left(\nicefrac{\subbeta}{\beta}\right) \, v_{t+1} < 0$, then by Lemma \ref{lem:prop-g}, $v)$, $g(w^\star) \geq g(\left(\nicefrac{\subbeta}{\beta}\right) \, v_{t+1}) $ and hence it follows from Equation \ref{eq_something} that $\left(\nicefrac{\subbeta}{\beta}\right) \, v_t < w^\star$. We conclude that $\left(\nicefrac{\subbeta}{\beta}\right) \, v_t < w^\star$ for all $t\in\{1,\ldots,T\}$.

Hence, since $\left(\nicefrac{\subbeta}{\beta}\right) \, v_{t+1} \leq w^\star$, we have
\[
\left(\nicefrac{\subbeta}{\beta}\right) \, v_{t+1} \leq g(\left(\nicefrac{\subbeta}{\beta}\right) \, v_{t+1}) = \left(\nicefrac{\subbeta}{\beta}\right) \, v_{t} \Rightarrow v_{t+1} \leq v_{t}\,.
\]

Finally, we establish the result for $\delta=1$. First, note that the right-hand-side of \eqref{eq:def-g} is continuous in $\delta$ and as $\left(\nicefrac{\subbeta}{\beta}\right)\,\,v_{t} = g\left( \left(\nicefrac{\subbeta}{\beta}\right)\,\,v_{t+1} \right)$ by Lemma \ref{lem:prop-g} $i)$, it follows that the continuation values $v_1,\ldots,v_T$ are continuous in $\delta$. Let $v_t^\delta$ be the continuation value in period $t$ as a function of $\delta$. We already established that $v_t^\delta - v_{t+1}^\delta \geq 0$ for all $\delta<1$. By continuity, we have that $v_t^1 - v_{t+1}^1 = \lim_{\delta \nearrow 1} v_t^\delta - v_{t+1}^\delta \geq 0$.
\qed

\bigskip

\noindent {\bf Proof of Theorem \ref{thm:non-identifiability-sophisticate}:} Fix a non-decreasing sequence of stopping probabilities $0 < p_1 \leq p_2 \leq \ldots \leq p_T < 1$, $(\delta,\beta) \in (0,1] \times (0,1]$, and a penalty $\nicefrac{\by}{\beta \delta} \in \RR$. We will construct a distribution $F$ that implies the stopping probabilities $p$ for a sophisticate.

Pick any perceived first-period cutoff $c_1 >0$ such that
\[
	c_1 > \max \left\{0, -(1 - \beta) \,\delta  \by  \,\,\frac{1-(\delta \, \frac{1-p_T}{2})^{T-1}}{(1-(\delta \, \frac{1-p_T}{2})) (\delta \, \frac{1-p_T}{2})^{T-1}} \right\}.
\]

Using $\subbeta=\beta$ in Lemma \ref{lem:rec-representation}, the perceived continuation values satisfy
\begin{equation}
	v_t = \begin{cases} \beta \,\delta \int_{v_{t+1}}^\infty z \,d F(z) + F(v_{t+1}) \, \delta \, v_{t+1} & \text{ for }t \in \{1,\ldots,T-1\}\\
	\by &\text{ for } t=T\end{cases}\,.
\end{equation}
Let $F$ be the sum of $T+2$ Dirac measures
\begin{align}\label{eq:def-F-sophisticate}
F(x; v) &= \sum_{k=0}^{T+1} f_k\,  \mathbf{1}_{\pi_k(v) \leq x}  ,
\end{align} 
at the mass points $\pi_0,\ldots,\pi_{T}$ satisfying
\begin{equation*}
\pi_k(v) = \begin{cases} \underline{y} - c_1 &\text{ if } k =0\\
\underline{y}  &\text{ if } k =1\\
v_{T-k+1}  &\text{ if } k  \in \{ 2,\ldots,T\}
\end{cases}\,.
\end{equation*}
Let the probability of each mass point be given by
\begin{equation*}
	f_k = \begin{cases} (1-p_T)/2 &\text{ if } k = 0,1\\
p_{T-k+2} - p_{T-k+1} &\text{ if } k \in \{ 2,\ldots,T\}\\
 p_1  &\text{ if } k = T+1
\end{cases}\,.
\end{equation*}
Note that $f_0>0$ as $p_T<1$. Since the mass points of $F$ are exactly at the continuation values, we get that for $t \in \{1,\ldots,T-1\}$ the recursive equation for the continuation values $v$ simplifies to a recursive equation for the mass points $\pi$; i.e.
\begin{align} 
	\pi_{T+1-t} &= \beta \,\delta \int_{v_{t+1}}^\infty z \,d F(z) + F(v_{t+1}) \, \delta \, v_{t+1} = \beta \,\delta \sum_{j=T-t+1}^{T+1} f_j \pi_j + \delta \, \left( \sum_{j=0}^{T-t} f_j \right)  \pi_{T-t}  \label{eq:cont_value_mass} \\
	\Rightarrow \pi_k &= \beta \,\delta \sum_{j=k}^{T+1} f_j \pi_j + \delta \,\left( \sum_{j=0}^{k-1} f_j \right) \pi_{k-1} \hspace{3mm}\text{ for } k \in \{2,\ldots,T\} \label{eq:pik}\,. 
\end{align}
We furthermore restrict attention to distributions for which Equation \eqref{eq:cont_value_mass} is also satisfied for $t=T$, i.e. for which $\pi_1$ satisfies Equation \eqref{eq:pik} evaluated at $k=1$. In that case, \eqref{eq:pik} implies that for $k \in \{2,\ldots,T\}$, 
\begin{align}\label{eq:pi-forward}
	\left( \pi_k - \pi_{k-1} \right) &= (1 - \beta) \,\delta f_{k-1} \pi_{k-1} + \delta \,\sum_{j=0}^{k-2} f_j  \left( \pi_{k-1} - \pi_{k-2} \right) \,.
\end{align}
As \eqref{eq:pi-forward} can be solved forward and $\pi_0,\pi_1$ are known, we can use it to determine $\pi_2 , \ldots, \pi_T$. Given the values $\pi_0, \ldots, \pi_T $, we can determine $\pi_{T+1}$ by solving \eqref{eq:pik} for $k=T$
\begin{align*}
	\pi_T &= \beta \delta (f_T \pi_T + f_{T+1} \pi_{T+1}) + \delta \pi_{T-1} \left( \sum_{j=0}^{T-1} f_j  \right)\,. 
\end{align*}
Denote this solution by $\pi^\star$. If $\pi^\star$ is strictly increasing then the distribution defined in \eqref{eq:def-F-sophisticate} has mass points exactly at the continuation values $v$ and leads to the given stopping probabilities $p$. 

We are thus left left to show that the resulting solution $\pi_0^\star, \pi_1^\star, \ldots, \pi_{T+1}^\star$ is increasing. We will show that $\pi_k^\star - \pi_{k-1}^\star >0$ by induction for $k \in \{1,\ldots,T\}$. $\pi_0^\star < \pi_1^\star$ by construction as $c_1 > 0$. We next do the induction step and assume that $\pi^\star_0 < \pi_1^\star < \ldots < \pi_{k-1}^\star$. Since for $k \geq 2$ one has $\pi^\star_{k-1} > \pi_1^\star = \by$,  \eqref{eq:pi-forward} implies that
\begin{align}
	\left( \pi_k - \pi_{k-1} \right) &\geq (1 - \beta) \,\delta  \by + \delta \, f_0  \left( \pi_{k-1} - \pi_{k-2} \right) \nonumber \\
	&= \alpha + \gamma \left( \pi_{k-1} - \pi_{k-2} \right)\,,
\end{align}
where $\alpha=(1 - \beta) \,\delta  \by$ and $\gamma = \delta \, f_0 \in (0,1)$. Since for $\by\geq 0$, we have  $\alpha \geq 0$, it follows that $\pi$ is non-decreasing in this case. We are left to show the result for $\by<0$ and $\alpha<0$. This implies that
\begin{align}
	\left( \pi_k^\star - \pi_{k-1}^\star \right) &\geq \alpha \sum_{j=0}^{k-2} \gamma^j + \gamma^{k-1} \left( \pi_1^\star - \pi_{0}^\star \right) \nonumber \\
	&= \alpha \frac{1-\gamma^{k-1}}{1-\gamma} + \gamma^{k-1} c_1 \nonumber \\
	&\geq \alpha \frac{1-\gamma^{T-1}}{1-\gamma} + \gamma^{T-1} c_1 \nonumber \\
	&= \gamma^{T-1} \left( c_1 - |\alpha| \frac{1-\gamma^{T-1}}{(1-\gamma) \gamma^{T-1}} \right) > 0\,.
\end{align}
The last inequality here follows from our choice of $c_1$. We thus have shown that $\pi_0^\star < \pi_1^\star < \ldots < \pi_T^\star$. It is left to show that $\pi_T^\star < \pi_{T+1}^\star$. By chosing $c_1$ large enough, we can without loss of generality assume that $\pi_T^\star>0$. If $\pi_T^\star > 0$ and $\pi_{T+1}^\star \leq \pi_T^\star$, we have that
\begin{align*}
	\pi_T^\star &= \beta \delta (f_T \pi_T^\star + f_{T+1} \pi_{T+1}^\star) + \delta \pi_{T-1}^\star \left( \sum_{j=0}^{T-1} f_j  \right)\\
	& \le \pi_T^\star \beta \delta f_T  + \pi_{T}^\star\beta \delta f_{T+1}  +  \pi_{T-1}^\star \delta  \left( \sum_{j=0}^{T-1} f_j  \right) \\
	\Leftrightarrow 1 & \le \beta \delta f_T + \beta \delta f_{T+1} + \frac{ \pi_{T-1}^\star}{\pi_{T}^\star} \delta  \left( \sum_{j=0}^{T-1} f_j  \right)\,.
\end{align*}
As $f_T+ f_{T+1}+ \left( \sum_{j=0}^{T-1} f_j  \right)=1$, $f_0 >0$, and $\pi_{T-1}^\star <\pi_{T}^\star$, this is a contradiction and completes the proof. \qed

\bigskip

\noindent {\bf Proof of  Lemma \ref{lem:aux-properties-naive}:}
\textbf{ $i)$:}  By Theorem \ref{prop:monotone-values}, the subjective continuation vales are weakly decreasing. For the sake of a contradiction, suppose the subjective continuation value is constant across two periods.
 Denote by $\underline{m}=\min \big(\text{supp}\, F \big)$ the left end-point of the support of $F$. By assumption $\underline{m}\leq \by   <  0$. By Lemma \ref{lem:prop-g} $i)$, we have that $\nicefrac{v_t}{\beta}=g(\nicefrac{v_{t+1}}{\beta})$ for all $t\in \{1,\ldots,T-1\}$, where,  by Equation \ref{eq:alt-representation}, $g(x)=\delta \int_{-\infty}^\infty \max\Big\{z,x\Big\} \,d F(z)$.  Note that $g$ is non-decreasing, strictly increasing for all $x\geq \underline{m}$, and that $g(x)=\delta \int_{-\infty}^\infty z \,d F(z)\geq \delta \underline{m} > \underline{m}$ for $x<\underline{m}$. Suppose that $v_{t-1}=v_{t}$ for some $t\in \{2,\ldots,T\}$. This implies that $\nicefrac{v_t}{\beta}=g(\nicefrac{v_{t+1}}{\beta})=g(\nicefrac{v_{t}}{\beta})$.  Hence, $\nicefrac{v_t}{\beta}  > \underline{m}$ and as $g$ is strictly increasing for $x\geq \underline{m}$, there can not exist a $\tilde{v}\neq v_t$ such that $\nicefrac{v_t}{\beta}=g(\nicefrac{\tilde{v}}{\beta})$. Hence,  $v_s = v_t$ for all $s,t \in \{1,\ldots,T\}$. As $v_T = \by$, this implies that $v_t = \by$ for all $t$. By Lemma \ref{lem:prop-g} $iii)$, however, any fixed point of $g$ is non-negative, so that $\bar{y}\geq 0$, contradicting the assumption that $\bar{y} < 0$.\\

\noindent We now show $ii)$: Let $v$ be the continuation values associated with $F$ and $\tilde{v}$ the continuation values associated with $\tilde{F} \prec_{FOSD} F$. We want to show that $v_t \geq \tilde{v}_t$ for every $t\in \{1,\ldots,T\}$. We show the result by backward induction over $T$. The start of the induction is that $v_T = \tilde{v}_T =\by$. In the induction step, we show that $v_{t+1}\geq \tilde{v}_{t+1}$ implies $v_t\geq \tilde{v}_t$ 
\begin{align*}
	\nicefrac{v_t}{\beta}&=\delta \int_{-\infty}^\infty \max\Big\{z,\nicefrac{v_{t+1}}{\beta}\Big\} \,d F(z) \geq  \delta \int_{-\infty}^\infty \max\Big\{z,\nicefrac{\tilde{v}_{t+1}}{\beta}\Big\} \,d F(z) \\
	&\geq \delta \int_{-\infty}^\infty \max\Big\{z,\nicefrac{\tilde{v}_{t+1}}{\beta}\Big\} \,d \tilde{F}(z) = \nicefrac{\tilde{v}_t}{\beta} \,.  \ \ \ \ \ \ \ \ \ \ \qed
\end{align*}

\noindent We are now ready to prove Theorem \ref{thm:non-identifiability}.

\bigskip

\noindent {\bf Proof of Theorem \ref{thm:non-identifiability}:}
Let $G_{a,b}(x) = \max \{ \min \left\{ \frac{x - a}{b-a}, 1 \right\} ,0 \}$ be the uniform CDF on $[a,b]$ for $a<b$ and a Dirac measure $G_{a,a}(x)=\mathbf{1}_{a\leq x}$ for $a=b$. Fix some $c_1,c_2 > 0$.
Consider a  non-decreasing sequence of  stopping probabilities $0 < p_1 \leq \ldots \leq p_T < 1$ and for every non-increasing sequence of continuation values $v_1 \geq \ldots \geq v_{T-1}$ with $ v_{T-1} \geq \underline{y}$, define the function $F$
\begin{align*}
F(x; v) &= \sum_{k=0}^T f_k\,  G_{\pi_k(v),\pi_{k+1}(v)} (x) ,
\end{align*}
where
\begin{equation*}
\pi_k(v) = \begin{cases} \underline{y} - c_1 &\text{ if } k =0\\
\underline{y}  &\text{ if } k =1\\
v_{T-k+1}  &\text{ if } k  \in \{ 2,\ldots,T\}\\
v_{1} + c_2  &\text{ if } k = T+1
\end{cases}\,,
\end{equation*}
and
\begin{equation*}
	f_k = \begin{cases} 1-p_T &\text{ if } k = 0\\
p_{T-k+1} - p_{T-k} &\text{ if } k \in \{ 1,\ldots,T-1\}\\
p_1  &\text{ if } k = T
\end{cases}\,.\
\end{equation*}

\noindent \textsl{$F$ is a distribution:} We begin by showing that $F$ is a cumulative distribution function. Note that $f_k\geq 0$ and that for $k<T$,
\begin{equation}\label{eq:sum-f}
\sum_{j=0}^k f_j = 1-p_{T-k}
\end{equation} and $\sum_{j=0}^T f_j = 1$. For every $v$, the function 
$F(\cdot;v)$ is non-decreasing and non-negative as the CDF $G$ is non-decreasing and non-negative. It thus follows that $F$ is a well defined CDF with support $[\pi_0,\pi_{T+1}]=[\underline{y}-c_1, v_1 + c_2]$.\\

\noindent \textsl{Continuation values induced by $F$:} Consider now the continuation values $w$ induced by $F(\cdot;v)$. By Lemma \ref{lem:rec-representation}, they solve the equation
\begin{equation}\label{eq:fixedpoint-representation}
	\frac{w_{t}}{\beta} =   \delta \int_{-\infty}^\infty \max\Big\{z,\frac{w_{t+1}}{\beta}\Big\} \,d F(z;v) \,\,\,\,\,  \text{ for }t \in \{1,\ldots,T-1\}\,,
\end{equation}
with $w_T=\by$. Denote by $L:\RR^{T-1} \to \RR^{T-1}$ the function mapping $v$ to $w$ using \eqref{eq:fixedpoint-representation}. By Theorem \ref{prop:monotone-values}, $w = L(v)$ is non-increasing. As $w$ is non-increasing and $w_T=\underline{y}$, it follows that $(Lv)_t \geq \underline{y}$ for all $t\in \{2,\ldots,T-1\}$. Furthermore, as $\text{supp} F(\cdot;v) \subseteq [\underline{y}-c_1, v_1 + c_2]$
\begin{align*}
	w_1 &= \beta \delta \int_{-\infty}^\infty \max\Big\{z,\frac{w_{t+1}}{\beta}\Big\} \,d F(z;v) \leq \beta \delta \int_{-\infty}^\infty \max\Big\{v_1+c_2,\frac{w_{1}}{\beta}\Big\} \,d F(z;v)\\
	&= \delta \beta \max\Big\{(v_1 + c_2), \frac{w_1}{\beta}\Big\} \leq \delta \beta (v_1 + c_2) \leq \delta (v_1 + c_2) \,.
\end{align*}
Thus, if $v_1 \leq \frac{\delta}{1-\delta} c_2$, we have that 
\[
w_t\leq w_1 \leq \delta  (v_1 + c_2) \leq \frac{\delta}{1-\delta} c_2 \,.
\]
Consequently, $L$ maps $M$ into itself, where $M$ is the set of non-increasing sequences contained in $[\underline{y},\frac{\delta}{1-\delta} c_2]^{T-1}$, i.e.
\[
M=\left\{m \in \Big[\underline{y},\frac{\delta}{1-\delta} c_2\Big]^{T-1}:  m_1 \geq m_2 \geq \ldots \geq m_{T-1} \right\}\,.
\]

\noindent \textsl{Any fixed-point of $L$ induces a solution:} We next argue that if $w^\star \in \RR^{T-1}$ is a fixed point of $L$ then the distribution $F(\cdot ;w^\star)$ induces the stopping probabilities $p$ and thus solves our problem. By Lemma \ref{lem:aux-properties-naive} $i)$, any fixed-point must be strictly decreasing 
$w_1^\star > w_2^\star > \ldots > w_{T-1}^\star$. As $w^\star$ is a fixed point of $L$, the agent stops in period $t$ if and only if $y_t \geq w_t^\star$, which happens with probability
\begin{align*}
\PP[ y > w_t^\star ] &= 1-F(w_t^\star;w^\star) = 1- \sum_{k=0}^T f_k\,  G_{\pi_k(w^\star),\pi_{k+1}(w^\star)} (w_t^\star) = 1- \sum_{k=0}^T f_k\,  \ind{\pi_{k+1}(w^\star)\leq w_t^\star} \\
&= 1- \sum_{k=1}^{T-1} f_k\,  \ind{w_{T-k}^\star \leq w_t^\star} - f_0\ind{\by-c_1 \leq w_t^\star} - f_T \ind{w_1^\star + c_2 \leq w_t^\star}\\
&= 1- \sum_{k=0}^{T-t} f_k\,  = 1- (1-p_t) = p_t \,.
\end{align*}
%
%
Where we used \eqref{eq:sum-f} in the second to last equality. Thus, any distribution associated with a fixed point of $L$ induces the correct stopping probabilities. \\

\noindent \textsl{$L$ has a fixed-point:} It remains for us to argue that $L$ has a fixed point. We note that $M$ is a complete bounded lattice, as the point-wise maximum (minimum) over increasing sequences is increasing.\footnote{To see this note, that $(\by,\ldots,\by)$ is a minimal element and $(\frac{\delta}{1-\delta}c_2,\ldots,\frac{\delta}{1-\delta}c_2)$ is a maximal element. Furthermore, the point-wise infimum and supremum over any subset of $M$ lie in $M$.} We next note that $F$ respects first order stochastic dominance (FOSD), ie. if $v \geqq w$ then $F(\cdot; v)$ is greater than $F(\cdot; w)$ in FOSD.\footnote{We use the notation $\geqq$ for point-wise comparisons.} By Lemma \ref{lem:aux-properties-naive} $ii)$, increasing the distribution of payoffs in FOSD will (weakly) increase the subjective continuation values.
As a consequence $L$ is a monotone operator, i.e. $L (v) \geqq L(w)$ if $v \geqq w$. By Tarski's fixed point theorem, $L$ thus has a fixed point on the lattice $M$.\\

\noindent \textbf{Uniqueness:} Finally, we note that as the subjective continuation values $w^\star$ is strictly decreasing $F(\cdot;w^\star)$ has no mass points. Consequently, the probability that the agent is ever indifferent between stopping and continuing equals zero. Thus, any perception perfect equilibrium leads to the same distribution $p$. \qed

\bigskip

\noindent {\bf Proof of  Lemma \ref{lem:mass_points_sufficient}:} Let the pair $u,G$ solve \ref{eq:constraints-sophisticate}. From now one, fix $u$. Let $\EE_G$ denote the expectation taken with respect to the cumulative distribution function $G$, and $\PP_G$ the probability mass with respect to $G$. 

We now specify a distribution $F$ that has the properties specified in the Lemma. The $T+1$ mass points $(\pi_0, \ldots,\pi_T)$ are located at
\begin{equation*}
	\pi_k = \begin{cases} \EE_G[y|y \leq v_T] &\text{ if } k = 0\\
\EE_G[y|v_{T-k+1} < y \leq v_{T-k}] &\text{ if } k \in \{ 1,\ldots,T-1\}\\
 \EE_G[y|v_1 < y]  &\text{ if } k = T
\end{cases}\,.\
\end{equation*}
and their probability mass is given by $f_k$ as specified in the Lemma. Observe that by construction, we have
$$
\pi_0 \leq v_T < \pi_1 \leq v_{T-1} <  \ldots \leq \pi_{T-1} \leq v_1 < \pi_T.
$$ 

Since $G$ solves \ref{eq:constraints-sophisticate} and $1 - F(v_t) = p_t$ for all $t\in \{1, \ldots,T \}$ by construction, one has 
$$
1 - F(v_t) = 1-G(v_t) \ \ \forall t \in \{1, \ldots, T\}.
$$
Furthermore,
\begin{align}
\nonumber
\int_{v_{t+1}}^\infty z \,d G(z) & = \sum_{k=T - t}^{T-1} \EE_G[y|v_{T-k+1} < y \leq v_{T-k}] \PP_G[y|v_{T-k+1} < y \leq v_{T-k}] +  \EE_G[y|v_1 < y]   \PP_G[y|v_1 < y] \\ \nonumber
&= \sum_{k=T - t}^{T} f_k \pi_k\\ \nonumber
& = \int_{v_{t+1}}^\infty z \,d F(z) \,. \nonumber
\end{align}
Thus, since $u,G$ solve \ref{eq:constraints-sophisticate} so do $u,F$. \qed

\bigskip

\noindent {\bf Proof of  Theorem \ref{thm:non-para_identification}:} Lemma \ref{lem:mass_points_sufficient} implies for a plausible data set that \eqref{eq:constraints-sophisticate} admits a solution if and only if there exists $\pi \in \RR^{T+1}, f \in \Delta^{T+1}$ and a monotone function $u$ such that
\begin{align}
	v_t &=u(m_t)  \ \  & \forall t \in \{ 1,\ldots, T\}\, ,\\
	\pi_0 \leq &v_T < \pi_1 \leq v_{T-1} <  \ldots \leq \pi_{T-1} \leq v_1 < \pi_T, \label{eq:order-pi}\\
	\sum_{k=T-t}^T\pi_k f_k  &=  \frac{\delta^{-1}\,v_t - (1-p_{t+1})  \, v_{t+1}}{\beta} \ \ \  & \forall t \in \{ 1,\ldots, T-1\}\, ,\label{eq:rec-eq-mass}\\
	\sum_{k=T-t+1}^T f_k &= p_t \ \ , & \forall t \in \{ 1,\ldots, T\}\label{eq:f-equal-p}\,.
\end{align}
Equation \ref{eq:f-equal-p} is equivalent to $f_T = p_1, f_0 = 1-p_T$ and for all $t \in \{2,\ldots,T\}$
\begin{align*}
	p_t - p_{t-1} = \sum_{k=T-t+1}^T f_k - \sum_{k=T-t+2}^T f_k = f_{T-t+1} \,,
\end{align*}
and thus completely determines $f$.
From now on we thus consider $f$ as given.

Equation \ref{eq:rec-eq-mass} for $t=1$ is equivalent to
\[
	\pi_{T-1} f_{T-1} + \pi_{T} f_{T}  =  \frac{\delta^{-1}\,v_1 - (1-p_{2})  \, v_{2}}{\beta} \,\,.
\]
We note that there exists $\pi$ satisfying the above equation and \eqref{eq:order-pi} if and only if
\begin{equation}
v_{2} f_{T-1} + v_1 f_{T}  < \frac{\delta^{-1}\,v_1 - (1-p_{2})  \, v_{2}}{\beta} \,\,.
\end{equation}
That this is necessary follows as \eqref{eq:order-pi} provides a lower bound on $\pi_{T-1}$ and $\pi_T$. Since, $f_T= p_1>0$, this is also sufficient as you can always chose $\pi_T$ arbitrarily large. Rearranging for $\beta$ and plugging in $f$ yields
\begin{equation}\label{eq:beta-at}
	\beta < \frac{\delta^{-1}\,v_1 - (1-p_{2})  \, v_{2}}{v_{2} (p_2 - p_1) + v_1 p_1} \,.
\end{equation}
Next, we consider \eqref{eq:rec-eq-mass} for $t \in \{2,\ldots,T-1\}$. Subtracting \eqref{eq:rec-eq-mass} evaluated at $t-1$ from \eqref{eq:rec-eq-mass} evaluated at $t$ yields
\begin{align*}
\pi_{T-t} f_{T-t} &= \sum_{k=T-t}^T\pi_k f_k - \sum_{k=T-t+1}^T\pi_k f_k =  \frac{\delta^{-1}\,v_t - (1-p_{t+1})  \, v_{t+1}}{\beta} - \frac{\delta^{-1}\,v_{t-1} - (1-p_{t})  \, v_{t}}{\beta}\, ,
\end{align*}
which is equivalent to
$$
\pi_{T-t} = \frac{v_{t+1} (p_{t+1} - p_t) - \delta^{-1} (v_{t-1} - v_t) + (1-p_t) (v_t - v_{t+1})}{\beta (p_{t+1} - p_t)}.
$$
The above equation admits a solution satisfying \eqref{eq:order-pi} if and only if for all $t \in \{2, \ldots, T-1\}$, $v_{t+1} <\pi_{T-t} \leq v_t$. Rewriting using the definition of $a(\delta, t)$ from the statement of the theorem, \ref{eq:rec-eq-mass} admits a solution satisfying \eqref{eq:order-pi} if for all $t \in \{2, \ldots, T-1\}$ both $ v_{t+1} \beta  <  v_{t+1} a(\delta,t) $ and $v_{t} \beta  \geq v_{t+1} a(\delta,t)$, and in addition
\begin{equation}\label{eq:beta-at}
	\beta < \frac{\delta^{-1}\,v_1 - (1-p_{2})  \, v_{2}}{v_{2} (p_2 - p_1) + v_1 p_1} \,. 
\end{equation} 
This completes the proof.\qed

\renewcommand{\baselinestretch}{1} \normalsize

\bibliography{../biblio/economics}
\bibliographystyle{../packages/aer}

\section*{Supplementary Appendix: Example \ref{ex:tax-sophisticate}. } 

This appendix establishes the claims made in Example \ref{ex:tax-sophisticate}. 

As the immediate payoffs of task-completion in the support of $F$ are always positive, an active agent will complete the task in the final period. Furthermore, observe that the agent will complete the task in the penultimate period even for the low benefit draw since 
 $$
 \frac{1}{4} > \beta \left\{\frac{19}{16}\right\};
 $$
 that is, the payoff of stopping immediately is greater than the discounted expected payoff of stopping in the final period. Hence, the subjective continuation value in period $t=1$ is the discounted expected value of always stopping in period 2, i.e.
 $$
 v_1 = \beta \left\{\frac{19}{16}\right\}.
 $$
 Given the continuation value $v_1$, the agent will complete the task in the first period for both a high and a low payoff realization. Hence,---for simplicity\footnote{Our payoff comparison between the case with and without a tax is unaffected by whether or not Self 0 applies $\beta$ to discount future payoffs.} following the usual convention and presuming the long-run Self 0 does not use $\beta$ to discount future payoff---Self 0's expected benefit from facing the stopping problem is $\nicefrac{19}{16}$.

 Suppose now that the agent faces a tax of (1/8) that she must pay when completing the task. Hence, she faces a new (after tax) payoff distribution $G$ that with probability $3/4$ pays $11/8$ and with probability $1/4$ pays $1/8$. Again, an active agent will complete the task in the final period as payoffs of doing so are always positive. Furthermore, when drawing the high payoff, an agent will always complete the task immediately.
 
In the penultimate period, the agent will not complete the task when drawing a low payoff of $1/8$ since
 $$
 \frac{1}{8} < \beta \left\{\frac{3}{4} \times \frac{11}{8} + \frac{1}{4} \times \frac{1}{8} \right\} = \frac{1}{8} \left\{ \frac{17}{16}\right\};
 $$
 that is the payoff of stopping immediately is lower than the discounted expected payoff of completing the task in the final period. 
 
Taking the behavior in the penultimate period into account, the agent's first-period continuation value is
 $$
 v_1^G = \beta \left\{ \frac{3}{4} \times \frac{11}{8} + \frac{1}{4} \times \frac{17}{16}\right\} = \beta \left\{\frac{83}{64} \right\} >\beta \left\{\frac{19}{16}\right\} = v_1.
 $$
This already establishes that the first period continuation value is higher with a tax than without it. Furthermore, when facing the tax, the agent will not complete the task in the first period when facing a low payoff since
$$
\frac{1}{8} < \beta \left\{\frac{83}{64} \right\} = \frac{1}{8} \left\{\frac{83}{64} \right\}.
$$
Using this fact, Self 0's expected payoff of facing the problem with a tax is 
$$
\left\{ \frac{3}{4} \times \frac{11}{8} + \frac{1}{4} \times \frac{83}{64}\right\} = \frac{347}{256} >\left\{ \frac{17}{16}\right\}. 
$$
We conclude that the sophisticated agent strictly prefers the tax.

\end{document}